\documentclass[11pt]{article}

\usepackage{amsmath,amssymb,amsfonts,amsthm}
\usepackage{graphicx}
\usepackage{textcomp}
\usepackage{authblk}
\usepackage{xcolor}
\usepackage{enumerate,subcaption}
\usepackage{algpseudocode}
\usepackage[ruled,vlined,resetcount,linesnumbered]{algorithm2e}
\usepackage{algcompatible}

\newtheorem{problem}{Problem}
\theoremstyle{definition}
\newtheorem{definition}{Definition}[section]
\usepackage{pgfplots}
\usepackage{tikz}
\usepackage{xspace}
\algnewcommand{\algorithmicand}{\textbf{and}}
\usepackage{makecell}
\usepackage{float}

\newcommand{\negl}{\mathrm{negl}}
\newcommand{\BS}[2]{#1.\mathrm{#2}}
\newcommand{\D}{\mathcal{D}}

\newcommand{\chain}{\mathrm{Chain}}
\newcommand{\History}{\mathrm{History}}

\newcommand{\LOG}{\mathrm{Ledger}}
\newcommand{\buffer}{\mathrm{Buffer}}
\newcommand{\cone}{\mathrm{Cone}}
\newcommand{\ProcessedWithTxCert}{\mathrm{ProcTxCertificate}}
\newcommand{\ProcessedTotalOrder}{\mathrm{ProcTotalOrder}}
\newcommand{\EqSet}{\mathrm{EqSet}}

\usepackage{hyperref}
\hypersetup{
    colorlinks=true,
    linkcolor=blue,
    filecolor=magenta,      
    urlcolor=cyan,
    pdftitle={Overleaf Example},
    pdfpagemode=FullScreen,
    }

\usetikzlibrary{fit,positioning,calc}
\usetikzlibrary{shadows,hobby}
\usetikzlibrary{fadings}
\usetikzlibrary{shapes.arrows,calc,quotes,babel}
\usetikzlibrary{graphs,graphs.standard,arrows.meta, shapes.misc, positioning,decorations.pathreplacing,calligraphy}

\SetAlgoLined
\SetNoFillComment
\DontPrintSemicolon
\SetKwInOut{Input}{input}
\SetKwInOut{Output}{output}

\SetCommentSty{mycommfont}
\SetKwComment{Comment}{/* }{ */}

\SetFuncSty{myprogfont}
\SetKwProg{Proc}{procedure}{:}{\KwRet}
\SetKwProg{Func}{function}{:}{\KwRet}
\SetKwProg{Loc}{Local variables}{:}{\KwRet}
\SetKwFunction{Sign}{Sign}
\SetKwFunction{CreateBlock}{CreateBlock}
\SetKwFunction{Hash}{Hash}
\SetKwFunction{Payload}{Payload}
\SetKwFunction{Tips}{Tips}
\SetKwFunction{Now}{Now}
\SetKwFunction{LastFinal}{LastFinal}
\SetKwFunction{CurrentCommit}{CurrentDigest}
\SetKwFunction{IsConflict}{IsConflict}
\SetKwFunction{IsValid}{IsValid}

\SetKwFunction{Commit}{Digest}
\SetKwFunction{UpdateChain}{UpdateChain}
\SetKwFunction{UpdateDAG}{UpdateDAG}
\SetKwFunction{FinalizeSlots}{FinalizeSlots}
\SetKwFunction{IsQuorum}{IsQuorum}
\SetKwFunction{SwitchChain}{SwitchChain}
\SetKwFunction{FinalizeTransactions}{ConfirmTxsConsensusPath}
\SetKwFunction{ConfirmTransactions}{ConfirmTxsFastPath}
\SetKwFunction{Leader}{Leader}
\SetKwFunction{AllTxCertificates}{AllTxCertificates}
\SetKwFunction{LastBlockFrom}{LastBlockFrom}
\SetKwFunction{LastCommitCertificate}{LastDigestCertificate}
\SetKwFunction{Broadcast}{Broadcast}
\SetKwFunction{Concat}{Concat}
\SetKwFunction{WakeUpChain}{WakeUpChain}
\SetKwFunction{ReachNumber}{ReachNumber}
\SetKwFunction{CheckELSS}{CheckELSS}
\SetKwFunction{BeforeCommit}{BeforeDigest}
\SetKwFunction{BeforeTime}{BeforeTime}
\SetKwFunction{FinalTime}{FinalTime}
\SetKwFunction{FinalTimes}{FinalTimes}
\SetKwFunction{Number}{Num}
\SetKwFunction{Mode}{Mode}
\SetKwFunction{Order}{Order}

\newif\ifcomment
\commenttrue

\title{Slipstream: Ebb-and-Flow Consensus on a DAG with Fast Confirmation for UTXO Transactions}
\author[1]{Nikita Polyanskii}
\author[2]{Sebastian M{\"u}ller}
\author[3]{Mayank Raikwar}
\affil[1]{IOTA Foundation}
\affil[2]{Aix Marseille Universit{\'e}}
\affil[3]{University of Oslo}
\date{}

\newtheorem{theorem}{Theorem}
\newtheorem{remark}{Remark}
\newtheorem{lemma}[theorem]{Lemma}



\begin{document}

\maketitle

\begin{abstract}
This paper introduces Slipstream, a Byzantine Fault Tolerance (BFT) protocol where nodes concurrently propose blocks to be added to a Directed Acyclic Graph (DAG) and aim to agree on block ordering. Slipstream offers two types of block orderings: an optimistic ordering, which is live and secure in a sleepy model under up to 50\% Byzantine nodes, and a final ordering, which is a prefix of the optimistic ordering and ensures safety and liveness in an eventual lock-step synchronous model under up to 33\% Byzantine nodes. Additionally, Slipstream integrates a payment system that allows for fast UTXO transaction confirmation independently of block ordering. Transactions are confirmed in three rounds during synchrony, and unconfirmed double spends are resolved in a novel way using the DAG structure.
\end{abstract}

\section{Introduction}\label{sec:intro}
The problem of ordering blocks, or more generally events, in a Byzantine distributed system, has been a fundamental challenge in distributed computing for decades~\cite{lamport2019byzantine}. Many practical solutions assume a partially synchronous network, where the communication is asynchronous for an unknown period but eventually becomes synchronous after a \textit{global stabilization time} (GST). Such protocols tolerate Byzantine faults if at least two-thirds of the nodes are correct.

Recently, DAG-based consensus protocols have emerged as a promising approach to distribute the responsibility for consensus among multiple nodes, achieving balanced network utilization and high throughput. 
In such protocols, nodes collectively build a Directed Acyclic Graph (DAG), where each vertex represents a block of transactions referencing previous blocks. By interpreting their local DAGs, nodes can determine a final block ordering.

However, existing DAG-based BFT protocols in partially synchronous networks lack \emph{dynamic availability}~\cite{neu2021ebb,pass2017sleepy}. Dynamic availability ensures that a protocol remains live even when portions of the network go offline. BFT protocols designed for safety in asynchronous settings cannot progress when more than one-third of the nodes are asleep, as a quorum cannot be achieved. Such protocols are live only after a \textit{global awake time} (GAT), when at least a supermajority of nodes get awake. 

To achieve the best of both worlds, a family of protocols, referred to as \textit{ebb-and-flow} protocols, was introduced in~\cite{neu2021ebb}. Such protocols produce two types of block ordering: optimistic or available block ordering favoring liveness under dynamic participation (and synchrony) and final block ordering favoring safety even under network partitions. This enables clients to rely on different ledgers based on their assumptions and requirements.

This paper presents \textbf{Slipstream}\footnote{The name Slipstream reflects the design where the final block ordering operates in the ‘slipstream’ of the optimistic block ordering, ensuring that while the optimistic ordering is always live, the final ordering follows closely behind to provide safety even under asynchrony.}, a first ebb-and-flow DAG-based consensus protocol. In addition, Slipstream implements a payment system for a more efficient UTXO transaction confirmation.

\subsection{Our contribution and models}
We consider two specific models for our protocol: the \textit{slot-sleepy} (SS) model and the \textit{eventual lock-step synchronous} (ELSS) model. In both models, the network consists of  $n$ nodes of which up to $f$ are \textit{Byzantine}, and at least $n-f$ are \textit{correct}. The execution of the protocol proceeds in \textit{rounds}, grouped into non-overlapping \textit{slots}, where each slot consists of  $f+2$ rounds.  Each round is identified by a timestamp  $\langle s,i\rangle$, $s\in \mathbb{N}, i\in\{1,\ldots,f+2\}$, where $s$ is the slot index and $i$ is the round index within the slot.

In the ELSS model, communication is asynchronous until GST, after which it becomes lock-step synchronous, meaning that a message sent by a correct node is received by other correct nodes in the next round. All nodes are awake, and the system operates with $n=3f+1$ nodes.

In the SS model, communication is always lock-step synchronous, but nodes may be asleep or awake in each slot; more specifically, each node is either awake or asleep for the whole slot. The number of correct awake nodes is strictly larger than the number of Byzantine awake nodes. The adversary can control who is awake. In this model, GST is set to 0, while the global awake time (GAT) is unknown. In Table~\ref{tab: models}, we highlight the main features of each of the models.

To facilitate synchronization in the ELSS model after GST, we assume a tiebreaker mechanism: at the first round of each slot, all correct nodes select the same correct leader with positive probability. 
This can be achieved using a random common coin, or verifiable random functions (VRFs)~\cite{micali1999verifiable}. 

\begin{table}[t]
    \centering
    \begin{tabular}{|c||c|c|c|c|c|}
    \hline
    \textbf{Model}  & GST & GAT & \makecell{Fraction of correct \\ among awake} & \makecell{Assumption \\ for tiebreaker} & \makecell{Properties achieved \\
    by Slipstream}\\
       \hline
       ELSS  & $\in \mathbb{N}$ & 0 &  $>2/3$ & Yes & \textbf{P2}-\textbf{P6}\\
       \hline
        SS & 0 & $\in \mathbb{N}\cup\{\infty\} $ &  $>1/2$& No & \textbf{P1},\textbf{P3}\\
        \hline
    \end{tabular}
    \caption{Eventual lock-step synchronous (ELSS) and slot-sleepy (SS) models}
    \label{tab: models}
\end{table}

Our main contribution is the design of \textbf{Slipstream}, a DAG-based BFT protocol that simultaneously satisfies the following key properties
(the links to formal statements are given in parentheses): \\
\textbf{P1:} The protocol provides an \textit{optimistic block ordering}, which is safe and live in the SS model. Irrespective of GST, a node proposing a block gets this block in its optimistic ordering with latency $O(f)$ (Th.~\ref{th: dynamically available}). \\
\textbf{P2:} The protocol provides a \textit{final block ordering}, which is a prefix of the optimistic block ordering and is safe and live in the ELSS model. After GST and a leader-based merging of local DAGs\footnote{If GST$>0$, successful merging happens on average with the latency of $O(nf)$ rounds.}, every block of a correct node gets in final ordering with latency $O(f)$ (Th.~\ref{th: eventually syncrony}). \\
\textbf{P3:} After GST and a leader-based merging of local DAGs, the protocol operates in a \textit{leaderless} way. \\
\textbf{P4:} The protocol provides a notion of \textit{confirmation} for UTXO transactions. This confirmation is safe and live in the ELSS model (Th.~\ref{th: confirmation}). \\
\textbf{P5:} After $\max\{$GST,GAT$\}$ and a leader-based merging of local DAGs, transactions get confirmed through a \textit{fast (consensusless) path} in three rounds (Th.~\ref{th: confirmation}).\\
\textbf{P6:} After $\max\{$GST,GAT$\}$ and a leader-based synchronization of local DAGs, transactions get confirmed through a \textit{consensus path} with latency $O(f)$  (Th.~\ref{th: confirmation}). In particular, the consensus path serves as an \textit{unlocking mechanism} that resolves unconfirmed double spends by utilizing the DAG structure only.

The amortized communication complexity to commit one bit of transactions in Slipstream is $O(n^2)$ (see Rem.~\ref{rem: on communication complexity}).

\subsection{Comparison with related work}
\textbf{Ebb-and-flow protocol:} The first two properties \textbf{P1-P2} imply that Slipstream is of the ebb-and-flow type as introduced in~\cite{neu2021ebb} and achieve an optimal Byzantine fault tolerance resilience. All existing ebb-and-flow protocols have subprotocols for a sleepy model, which are leader-based~\cite{d2022goldfish,d2023simple,neu2021ebb}. Since randomization is needed to select the leader, the adversary always has a chance to guess the leader and make the leader go to sleep. This implies that these protocols achieve almost-surely termination, i.e., the liveness of optimistic ordering is probabilistic. Moreover, this issue was posed as an open question~\cite{momose2022constant} for all sleepy protocols (not necessarily ebb-and-flow) whether probabilistic termination is inherent or avoidable by leaderless protocols. We partially answer this question as the termination in Slipstream is achieved deterministically in the SS model and does not rely on randomization. However, we note that our SS model has stronger requirements than many other sleepy models: 1) the SS model has a lock-step nature, whereas $\Delta$-synchronous communication is considered in some sleepy models; 2) correct awake nodes are awake for the whole slot (consisting of $f+2$ rounds) in the SS model, whereas such stable participation is usually assumed for time $c \Delta$ in other sleepy models, e.g., $c=7$ in~\cite{momose2022constant}. Table~\ref{tab: ebb-and-flow protocol} compares Slipstream with two other ebb-and-flow protocols (we assume delay $\Delta=1$ in other protocols). The latency to settle a block in the optimistic block ordering with probability $1$ of Slipstream is better than in other protocols. The happy-case (e.g., when the selected leader is correct and awake) latency in an eventually synchronous model of the propose-vote-merge protocol~\cite{d2023simple} and the snap-and-chat protocol~\cite{neu2021ebb} (with a sleepy subprotocol from~\cite{ConstantLatencyinSleepyConsensus}) is $O(1)$ and could be as low as $4$ rounds. This is worse than $3$ rounds for UTXO transactions in Slipstream, but better than the latency of $O(f)$ rounds for other types of transactions. Since Slipstream is a DAG-based protocol, it settles the order over all blocks proposed in one slot, which could be $\Theta(fn)$, in one shot with latency $O(f)$. This is not the case for chain-based protocols, which settle, on average, at most one block in a slot.

\begin{table}[t]
    \centering
    \begin{tabular}{|c||c|c|c|c|c|c|}
    \hline
    \textbf{Protocol}  & \makecell{Latency in \\ sleepy model} & \makecell{Leaderless} &  \makecell{Extra latency \\ in ES model}  & \makecell{Fastest \\ latency} & \makecell{DAG/chain \\ -based}\\ 
       \hline
       Slipstream & $O(f)$ rounds & Yes & $+3$ rounds& $3$ rounds & DAG \\
        \hline
       Snap-and-chat~\cite{neu2021ebb} & \makecell{$\ell$ with prob.\\
       $1-1/\negl(\ell)$}  & No & \makecell{$+3$ phases of \\ Streamlet} & $>4$ rounds &chain \\
        \hline
      \makecell{Propose-vote-merge\\ \cite{d2023simple}} & \makecell{$\ell$ with prob.\\
       $1-1/\negl(\ell)$}  & No & $+3$ rounds & 4 rounds& chain \\       
        \hline
        
    \end{tabular}
    \caption{Comparison with ebb-and-flow protocols. We provide the latency in a sleepy model for a block to appear in an optimistic ordering and extra latency for a block from the optimistic ordering to appear in a final ordering. In addition, we provide the fastest latency in a happy case to commit any type of transaction. Function $\negl(\ell)$ grows faster than any polynomial.}
    \label{tab: ebb-and-flow protocol}
\end{table}
\textbf{Payment system:} The last three properties \textbf{P4-P6} are about the payment system integrated in Slipstream. The ordering of transactions can be derived from a block ordering, however, it is shown in~\cite{guerraoui2019consensus} that a system that enables participants to make simple payments from one account to another needs to solve a simpler task. In cases where payments are independent of one another (e.g., single-owned token assets or UTXO transactions), a total ordering becomes unnecessary, and a partial ordering suffices. This was later observed in multiple papers~\cite{baudet2020fastpay,baudet2023zef}, and different solutions that utilize DAGs were suggested. 
In Flash~\cite{lewis2023flash}, nodes create blocks to approve transactions included in prior blocks, and a quorum of approvals is sufficient to commit a UTXO transaction.
This solution achieves very low communication complexity and latency even in asynchronous settings, i.e., the latency can be as low as two network trips compared to three in \textbf{P5}. However, this and some other payment systems didn't consider a practical challenge in which unconfirmed double spends make funds forever locked. This challenge is addressed by our consensus path of confirmation in \textbf{P6}.

\textbf{DAG-based BFT protocol with integrated payment system:} 
Now we compare Slipstream with the most closely related works, Sui-Lutris~\cite{blackshear2023sui} and Mysticeti-FPC~\cite{babel2023mysticeti}, as they are both DAG-based BFT protocols with integrated payment systems. These protocols do not make progress of block ordering in a sleepy model (compared to \textbf{P1}). They operate in an eventual synchronous network and provide 1) a final block ordering which is live and safe under up to $33\%$ of Byzantine nodes (similar to \textbf{P2}, however, their latency after GST and GAT is $O(1)$ compared to $O(f)$ in Slipstream), and 2) a fast confirmation for owned-object transactions (similar to \textbf{P4}); specifically, confirmation is achieved in three rounds (similar to \textbf{P5}) and owned objects that get locked in one epoch by uncertified double spends are unlocked during epoch reconfiguration and available to clients for issuing transactions in the next epoch. The last property serves the same purpose as \textbf{P6}; however, we claim that our solution has several advantages. 

Sui-Lutris has a subprotocol that lets nodes confirm owned-object transactions \textit{before} including them in consensus blocks by constructing explicit transaction certificates. This solution requires extra signature generation and verification for each transaction. In both Mysticeti-FPC and Slipstream, confirmation of owned-object and UTXO transactions happens \textit{after} including them in consensus blocks by interpreting the local DAGs. In Mysticeti-FPC, nodes explicitly vote for causally past transactions in their blocks. A block is called a certificate for a transaction if its causal history contains a quorum of blocks voting for this transaction. Certificates provide safety for confirmation as no two conflicting transactions could get certificates. Once a quorum of certificates for a transaction appears in the local DAG, a node can confirm this transaction through the \textit{fast path}. This might happen before a block containing the transaction is committed. If a block, which is a certificate for a transaction, is committed by the Mysticeti consensus, then the transaction can be confirmed by the \textit{consensus path}. In case none of the double spends gets a quorum of approvals, the corresponding owned objects get locked and can not be spent until the epoch reconfiguration happens. 

To mitigate the lengthy waiting period for unlocking funds, we modify the voting rule and consensus path confirmation. In Slipstream, blocks can vote only for transactions included in the current or previous slot. A quorum of certificates enables fast-path confirmation. For the consensus path, we leverage two facts: 1) blocks vote for transactions only in the recent past, and 2) a quorum is required for both finalizing block ordering and constructing transaction certificates. Once a slot is double finalized, we apply the consensus path to transactions from two slots earlier, since no new certificates will be issued for them. First, we confirm transactions with certificates in the final block ordering, then apply a total order to confirm non-conflicting remaining transactions. We demonstrate how a similar approach can reduce the waiting period for unlocking funds and resolving double spends in Mysticeti-FPC from one epoch (24 hours) to just a few rounds (seconds), as detailed in Appendix~\ref{sec: adapting Slipstream to Mysticeti-FPC}.

\begin{table}[t]
    \centering
    \begin{tabular}{|c||c|c|c|}
    \hline
    \textbf{Protocol}  & \makecell{Latency \\  UTXO txs} & \makecell{Latency \\ other txs} &  \makecell{Unlocking \\ waiting time} \\ 
       \hline
       Slipstream & $3$ rounds & $O(f)$ rounds & $O(f)$ rounds \\
        \hline
       Mysticeti-FPC~\cite{babel2023mysticeti} & \makecell{$3$ rounds}  & $O(1)$ rounds & one epoch (24 hours)  \\
        \hline
      \makecell{Flash~\cite{lewis2023flash}+Cordial Miners~\cite{keidar2022cordial}} & \makecell{$2$ rounds}  &$O(1)$ rounds & $\infty$  \\       
        \hline
        
    \end{tabular}
    \caption{DAG-based BFT protocols with payment systems.}
    \label{tab: payment system}
\end{table}

\subsection{Overview of Slipstream}
During each round, awake nodes propose their blocks of transactions to be added to the DAG, referencing earlier blocks in their local DAGs.   The key component of Slipstream is based on a technique not found in existing DAG-based BFT protocols resilient to asynchrony:

\textbf{Slot-based digest of a DAG}: At the end of each slot, each node computes a \emph{slot digest} summarizing the ordering of \textit{all} the blocks in the local DAG that are created in previous slots (see Fig.~\ref{fig: commitments}). In addition, the slot digest commits to the digest of the slot before, thereby forming a \textit{backbone chain} of slot digests. Recall that in most existing DAG-based BFT protocols~\cite{babel2023mysticeti,danezis2022narwhal,gkagol2018aleph,malkhi2023bbca,shrestha2024sailfish,spiegelman2022bullshark}, digests or commitments are built using the causal history of leader blocks, where the leader nodes are determined by a predefined scheduler (or retrospectively using a random coin). The backbone chain is then formed using the chain of committed leader blocks. However, this approach does not work properly if one applies it directly to a sleepy model, as the adversary might make the leaders in the scheduler asleep, disabling awake nodes from creating digests and progressing the block ordering.

A slot digest and the corresponding backbone chain correspond to a unique optimistic block ordering. Each node adopts one digest for the whole slot (e.g., the node includes the digest in its blocks), updates the backbone chain at the end of the slot, and potentially can switch its backbone chain at the first round of the next slot. Nodes that were asleep at the previous slot join the protocol by adopting the digest of a majority of nodes from the previous slot. See Fig.~\ref{fig: overview} for an overview of the main steps.
\begin{figure}[t]
\begin{center}
\begin{tikzpicture}[scale=0.8, transform shape]
    \draw[thick] (-4,-1) -- (4,-1);
    \draw[thick, dashed ] (4,-1) -- (12,-1);
    
    \draw[thick] (12,-1) -- (16,-1);

    \node at (-2, 1) {\textbf{Propose}};
    \node[below] at (-2, 0.7) {\small Node proposes} ;
    \node[below] at (-2, 0.3) {\small its  digest for};
    \node[below] at (-2, -0.1) {\small slot $s-1$};
    \draw[dashed] (-2,-0.6) -- (-2,-1.5);

    \node at (2, 1) { \textbf{Merge and Vote}};
    \node[below] at (2, 0.7) {\small Node might merge};
    \node[below] at (2, 0.3) {\small local  DAGs and};
    \node[below] at (2, -0.1) {\small votes for one digest};
    \draw[dashed] (2, -.6) -- (2,-1.5);

    \node at (6, 1) { \textbf{Certify}};
    \node[below] at (6, 0.7) {\small Node might};
    \node[below] at (6, 0.3) {\small certify one digest};
      \node[below] at (6, -0.1) {\small and might finalize it};

    \node at (10, 1) { \textbf{Synchronization}};
    \node[below] at (10, 0.7) {\small Awake nodes with same digest} ;
    \node[below] at (10, 0.3) {\small synchronize prefixes of};
    \node[below] at (10, -0.1) {\small their DAGs up to slot $s$};
    \node[below] at (10, -0.2) { };

    \node at (14, 1) { \textbf{Propose}};
    \node[below] at (14, 0.7) {\small Node proposes} ;
    \node[below] at (14, 0.3) {\small its  digest for };
    \node[below] at (14, -0.1) {\small slot $s$};
    \node[below] at (14, -0.2) { };
    \draw[dashed] (14,-0.6) -- (14,-1.5);

    \node[below] at (-2, -1.5) {$\langle s, f+2 \rangle$};
    \node[below] at (2, -1.5) {$\langle s+1, 1 \rangle$};
    \node[below] at (8, -1.5) {$\ldots\ldots\ldots\ldots$};
    \node[below] at (14, -1.5) {$\langle s+1, f+2 \rangle$};
\end{tikzpicture}
\end{center}
\caption{Different phases of Slipstream within a slot}
\label{fig: overview}
\end{figure}

By induction, the protocol can be shown to solve the consensus problem in a sleepy model if the following property holds: if two correct nodes adopt the same digest when entering a slot, then they will generate the same digests at the end of the slot. To this end, we apply the idea of the Dolev-Strong Byzantine Agreement (BA)~\cite{dolev1983authenticated} to the DAG setting. The DAG structure allows to amortize the signature computations and verifications required by the Dolev-Strong BA, and, surprisingly, this use case was not described in the literature before:

\textbf{Synchronization:} Before adding a new block $B$ from the current slot in a local DAG, the node first ensures that $B$ contains the same digest as the node adopts. Then the node checks all \textit{unvalidated} blocks in the causal history of $B$, that are not yet in the local DAG and created at the previous slots.  Specifically, for each unvalidated block, one checks all the blocks from the current slot, from which one can observe (or reach) that unvalidated block (by traversing the DAG). If for each unvalidated block, the number of observers is at least the round offset index within the slot (corresponding to the time when the check is performed), the block $B$ with its history can be safely added to the local DAG.

While Slipstream under synchrony already outputs the optimistic block ordering,  one needs to have a DAG synchronization mechanism to capture the ELSS model for the time when asynchrony ends. One natural solution is to adopt the optimistic block ordering from a randomly chosen ``leader'' and merge the DAG from the leader with its own. However, this brings another challenge as following the leader might break safety for the optimistic block ordering. Before digging into that rule we explain our mechanism for final block ordering, which has similarities with PBFT-style consensus protocols and has even more parallels with Cordial Miners~\cite{keidar2022cordial} and Mysticeti~\cite{babel2023mysticeti}, two other DAG-based BFT protocols working with uncertified DAGs. Nodes in Slipstream examine their local DAGs to find blocks serving as implicit certificates for slot digests, whereas nodes in Cordial Miners and Mysticeti find blocks which are certificates for leader blocks. 

\textbf{Finality and digest certificates:} A block is said to be a \textit{certificate} for a slot digest if one can reach from this block a quorum of blocks in the same slot which include the same slot digest. Once the local DAG contains a quorum of certificates, the slot digest and the corresponding block ordering are finalized. 

The idea of introducing certificates is to provide safety of the final block ordering. Specifically, if one correct node finalizes a slot digest, then a majority of correct nodes has created certificates for the same slot digest. Nodes with certificates get locked and cannot switch their backbone chains to a one diverging before the certified digest unless extra conditions holds. Next we explain how DAG merging is performed and which conditions allow a locked node to switch its backbone chain.

\textbf{Leader-based merging and switching backbone chains:} We assume that after GST, at the beginning of each slot, any correct node can be selected as a leader in the perception of all correct nodes with a positive probability. If an awake node enters a new slot, it might try to switch its backbone chain to the one of the leader.  On a high level, a node should \textit{not} switch its backbone chain when one of the following conditions holds: 
\begin{itemize}
    \item the node has generated the same digest as a majority of nodes from the last slot. This condition allows correct nodes to stay on the same backbone chain in the SS model and not be disrupted by the adversary.
    \item the node has finalized a digest at the previous slot. This condition allows correct nodes staying on the same backbone chain after being synced in the ELSS model and not being disrupted by the adversary.
    \item the latest digest certificate that the leader holds, was created before the node's certificate. This ensures the safety of final block ordering.
\end{itemize}
In other cases (with several extra conditions ensuring the liveness of the protocol), the node has to switch its backbone chain to the one of the leader and merge two DAGs: its own and the one of the leader. With a positive probability, all correct nodes can be shown to adopt the same backbone chain, merge their DAGs, and continue progressing the final block ordering.

\subsection{Outline}
The remainder of the paper is organized as follows. Sec.~\ref{sec: system model} introduces formal network and communication models discussed in the paper and formulate problem statements for the respected models. Sec.~\ref{sec: preliminaries} discusses several notations and definitions needed to describe the protocol. We present our protocol Slipstream in Sec.~\ref{sec: protocol}. Sec.~\ref{sec: results} describes formal results achieved by Slipstream, which are then proved in Sec.~\ref{sec: proofs}. Appendix~\ref{sec: valid DAG} formally introduces the concept of a valid DAG in Slipstream. Most formal algorithms are deferred to Appendix~\ref{sec: algorithms}. Appendix~\ref{sec: comparison with dag-based bft protocls} gives a comparison of Slipstream with DAG-based BFT protocols. Appendix~\ref{sec: adapting Slipstream to Mysticeti-FPC} explains how to adapt the Slipstream consensus path for UTXO transactions to Mysticeti-FPC.

\section{System model and problem statements}\label{sec: system model}

The network consists of \( n \) nodes, where up to \( f \) nodes are Byzantine, and the remaining correct (honest) nodes follow the protocol. We consider two configurations for the network: \( n=3f+1 \) or \( n=2f+1 \), depending on the model used. Each node is equipped with a unique public-private key pair, where the public key serves as the node’s identifier. Nodes use their private key to sign blocks and the signatures are verifiable by other nodes using the corresponding public key. 

Time is divided into discrete \textit{slots}, each consisting of \( f+2 \) \textit{rounds}, indexed by a timestamp \(\langle s,i \rangle\), where \( s \) represents the slot number, and \( i \) the round within the slot.
The execution proceeds in three phases during each round: \textsc{Receive}, \textsc{State Update}, and \textsc{Send}. 

\subsection{Network and communication models}\label{sec: models}
We focus on two models in the paper: an \textit{eventual lock-step model} and a \textit{slot-sleepy model}. Before explaining them, we recall a standard lock-step synchronous model.
\begin{definition} [Lock-step synchronous model] A consensus protocol works under a lock-step synchronous model if and only if an execution consists of a sequence of rounds, and an round includes a sequence of the three phases:
\begin{itemize}
    \item A \textsc{Receive phase}, in which nodes receive messages sent by other nodes. In the current round, a node receives all messages from correct nodes that are sent in the previous round.
    \item A \textsc{State update phase}, in which a node takes an action based on the received messages from the receive phase.
    \item A \textsc{Send phase}, in which nodes send messages to other nodes.
\end{itemize}    
\end{definition}
\begin{remark}
Compared to the original definition of a lock-step synchronous model \cite{lock-step94}, we changed the order of the phases by moving the send phase to the end. This is done intentionally in order to simplify the description of a slot-sleepy model (see Def.~\ref{def: slot-sleepy model}). In particular, for the described sequence of phases, nodes can skip the receive and state update phases at the very first round in the execution of the protocol.
\end{remark}
In an eventual lock-step synchronous model, the network behaves as a lock-step synchronous model after an (unknown) finite asynchronous period ends.
\begin{definition}[Eventual lock-step synchronous model]\label{def:ELSS} A consensus protocol works under an eventual lock-step synchronous (\textup{ELSS}) model if there is a Global Stabilization Time, unknown to nodes and denoted as round $GST$, such that after $GST$ the protocol follows a lock-step synchronous model. In addition, 
\begin{itemize}
    \item Before round $GST$, nodes also have three phases for each round, but messages from the correct nodes can be delayed arbitrarily, i.e., in the receive phase of a round, messages from any prior round are not necessarily received.
    \item In the received phase of round $GST$, correct nodes receive all messages from the correct nodes that were sent at or before $GST$.
\end{itemize}
\end{definition}
In addition, we assume that after $GST$ with a positive probability, any correct node can be selected as a leader by all correct nodes when calling a function $\Leader()$.
\begin{definition}[Leader selection] \label{def: leader function}
A leader selection mechanism is a function $\Leader()$ that, given a slot number $s$, selects a leader for that slot with the following requirement. 
Denote the output of $\Leader()$ for node $i$ at slot $s$ as a random variable $\mathcal{L}_i^s$, then for $s \geq GST$ and for any honest node $j$, 
\[
\mathbb{P}(\mathcal{L}_i^s = j | \text{$i$ is honest}) \geq 1/n.
\]
\end{definition}

Two common ways to implement such a leader selection mechanism are:
\begin{itemize}
    \item Random Common Coin: The oracle takes the slot number $s$ as input and sends a random value $r(s)$ to all nodes during the send phase of round $\langle s-1, f+2\rangle$. This value is uniformly distributed and unpredictable, ensuring no node knows the leader before the receive phase of round $\langle s, 1\rangle$. For ease of presentation and proofs, we will assume in the proof the properties of a leader selection mechanism stemming from this mechanism.
    
    \item Verifiable Random Function (VRF): Each node computes a verifiable random output based on the slot number $s$ and its secret key. It shares this output during the send phase of $\langle s-1, f+2\rangle$. The node with the lowest VRF output is selected as the leader. Let us note that this method does not provide a random number that is unpredictable in the strict sense, as a malicious node with a low VRF output can predict an increased likelihood of being selected as leader. However, this does not reveal any information on which honest node will have the lowest VRF among the honest nodes. We refer to \cite{ConstantLatencyinSleepyConsensus} for more details and previous use of VRFs. 
\end{itemize}

In a slot-sleepy model, each node goes through sleep-wake cycles each of which corresponds to a sequence of slots.
\begin{definition}[Slot-sleepy model]\label{def: slot-sleepy model}
    A consensus protocol works under a slot-sleepy (\textup{SS}) model if, for each slot $s\in \mathbb{N}$, each node is either slot-$s$ awake or slot-$s$ asleep. If a node is slot-$s$ awake, then the node has all three (receive, state update, and send) phases at every round of slot $s$.
 If a node is slot-$s$ asleep, then the node does not have any phase at any round of slot $s$.
\end{definition}

Recall that the difference between these two models is highlighted in Table~\ref{tab: models}. 

\subsection{Problem statements}\label{sec: problem statements}
\subsubsection{Two problem statements for the \textup{ELSS} model} The first problem for this model is a classical consensus problem. We consider a notion of \textit{finality} for blocks,  which allows nodes to linearly sequence final blocks.  Each node maintains its final order $\Order_{\text{final}}$ of blocks.

\begin{problem}\label{problem: finality-final_blocks}
Design a DAG-based consensus protocol that satisfies two requirements.
\textbf{Safety:}  For the sequences of final blocks $\Order_{\textup{final}}$ by any two correct nodes, one must be the prefix of another. \textbf{Liveness:} If a correct node creates a block, then the block becomes eventually final for any correct node, i.e., it gets included to $\Order_{\textup{final}}$.
\end{problem}
The second problem is related to payment systems.  We assume that there exist accounts in the network. These accounts can issue \text{UTXO} transactions and broadcast them to nodes.   When creating blocks, nodes include a subset of transactions in their blocks. 

A \text{UTXO} is a pair consisting of a positive value and an account that can consume this \text{UTXO}. A single-owned \text{UTXO} transaction is a tuple containing (i) \text{UTXO} inputs, which can be consumed by one account, (ii) the account signature, and (iii) \text{UTXO} outputs. The sum of values in the inputs and the outputs are the same.
Two different \text{UTXO} transactions are called a \textit{double-spend} if they attempt to consume the same \text{UTXO}. Note that two identical transactions are not treated as a double-spend and generate different output. Each account begins with an initial \text{UTXO}, which specifies the initial balance of the account. 

The order over final blocks naturally allows to find the order over \text{UTXO} transactions in these blocks and check which of them can be executed, i.e., being confirmed. However, \text{UTXO} transactions don't require a total order, and \textit{confirmation} can be achieved through a consensusless path. We distinguish two paths how confirmation can be achieved: \textit{fast-path} confirmation and \textit{consensus-path} confirmation. In either case, if any correct node regards a \text{UTXO} transaction
as confirmed, then all correct nodes eventually regard it as confirmed. 
Once a \text{UTXO} transaction is confirmed, all its output \text{UTXO}s are said to get confirmed. 

We say that a transaction is \textit{cautious} and \textit{honest} if created by honest account, i.e., never creates a double-spend, and the inputs of the transaction are already confirmed in the perception of a correct node to which the transaction is sent.
\begin{definition}[Cautious account]\label{def: cautious account} A cautious account is an account that creates a \text{UTXO} transaction only if all the inputs to transactions are already confirmed by one correct node. Transactions issued by cautious accounts are called cautious transactions and cautious transactions are supposed to be sent to at least one such correct node. 
\end{definition}
\begin{remark}
In a Byzantine environment, it is in general not known which nodes are correct and which are not. However, having at least $f+1$ nodes confirming a transaction is sufficient for knowing that at least one correct node is confirming the transaction. Then the transaction can be sent to these $f+1$ nodes.  
\end{remark}
\begin{remark}
In the protocol itself, an account is not necessarily running a node. However, for simplicity, we can assume either an account has its own node, or a node is able to produce a proof of the transaction confirmation that can be verified by the account.
\end{remark}
\begin{definition}[Honest account]\label{def: honest account} An honest account is an account that uses any given \text{UTXO} associated with this account only once to create a \text{UTXO} transaction.  Transactions issued by honest accounts are called honest transactions.   
\end{definition}


Since we concentrate only on \text{UTXO} transactions in this paper, we consider a \text{UTXO} ledger, written as $\LOG$. This ledger is not a linearly ordered log of transactions, but instead, it is a set of confirmed UTXO transactions that must satisfy the two properties:\\
\textbf{1.} No double-spend appears in $\LOG$;\\
\textbf{2.} For any given UTXO transaction in $\LOG$, the transactions creating the inputs of the given transaction must be in $\LOG$.\footnote{We assume that the initial UTXOs of accounts are generated by the genesis transaction, which is always a part of the ledger $\LOG$.}
\begin{problem}
\label{problem: confirmation and finality}
Integrate a payment system in a DAG-based consensus protocol that satisfies the three requirements.
\textbf{Safety:} For any two \text{UTXO} transactions that constitute a double-spend, at most one of them could be confirmed, i.e., included in $\LOG$. \textbf{Liveness:} Every honest cautious \text{UTXO} transaction is eventually consensus-path confirmed, i.e., gets included in $\LOG$. There is a round $i$ after GST such that every honest cautious \text{UTXO} transaction created after $i$ is fast-path confirmed two rounds after including in a block by a correct node. \textbf{Consistency:} If a \text{UTXO} transaction is fast-path confirmed by one correct node, then it gets consensus-path confirmed by every correct node.
\end{problem}

\subsubsection{One problem statement for the \textup{SS} model}
We consider slot digests and a notion of \textit{commitment}, which allow nodes to linear sequence blocks committed to a slot digest. Each node maintains its optimistic order $\Order_{\text{own}}$ of blocks. 
\begin{problem}
\label{problem: finality-committed_blocks}
Design a DAG-based consensus protocol that satisfies two requirements.
\textbf{Safety:}  For the sequences of committed blocks $\Order_{\textup{own}}$ by any two correct nodes, one must be the prefix of another. \textbf{Liveness:} If a slot-$s$ awake correct node creates a block, then the block becomes committed by any slot-$t$ awake correct node with $t\ge s+2$.
\end{problem}

\section{Protocol preliminaries}\label{sec: preliminaries}
The basic unit of data in the protocol is called a \textit{block}, and nodes send only blocks as their messages.
    Each block $B$ contains: 
 \begin{itemize}
     \item \( \BS{B}{refs}\)  -- a list of hash references to previous blocks;
     \item \( \BS{B}{digest} \) -- a slot digest (see Def.~\ref{def: commitment});
     \item \( \BS{B}{txs} \) -- a set of transactions; 
    \item \( \BS{B}{time} \) -- a timestamp (for simplicity, we refer to the slot and round indices of the timestamp as $\BS{B}{slot}$ and $\BS{B}{round}$);
     \item \( \BS{B}{node} \) -- a node's unique identifier;
     \item $\BS{B}{sign}$ -- a node's signature that verifies the content of the block.
 \end{itemize} 

\begin{definition}[Reachable block]\label{def: reachable block and tx}
Block $B$ is called \textit{reachable} from block $C$ in a DAG if one can reach $B$ starting at $C$ and traversing the DAG along the directed edges. 
\end{definition}
In particular, $B$ 
is reachable from $B$. A special case of a DAG that we consider in the paper is defined based on the causal history of a block; see Fig.~\ref{fig: pastCone and Equivocation}.
\begin{definition}[Past cone]
    The past cone of block $B$, written as $\cone(B)$, is the DAG whose vertex set is the set of all blocks reachable from $B$. 
\end{definition}
\begin{figure}
\begin{subfigure}{.5\textwidth}
    \centering
    \includegraphics[width=.95\linewidth]{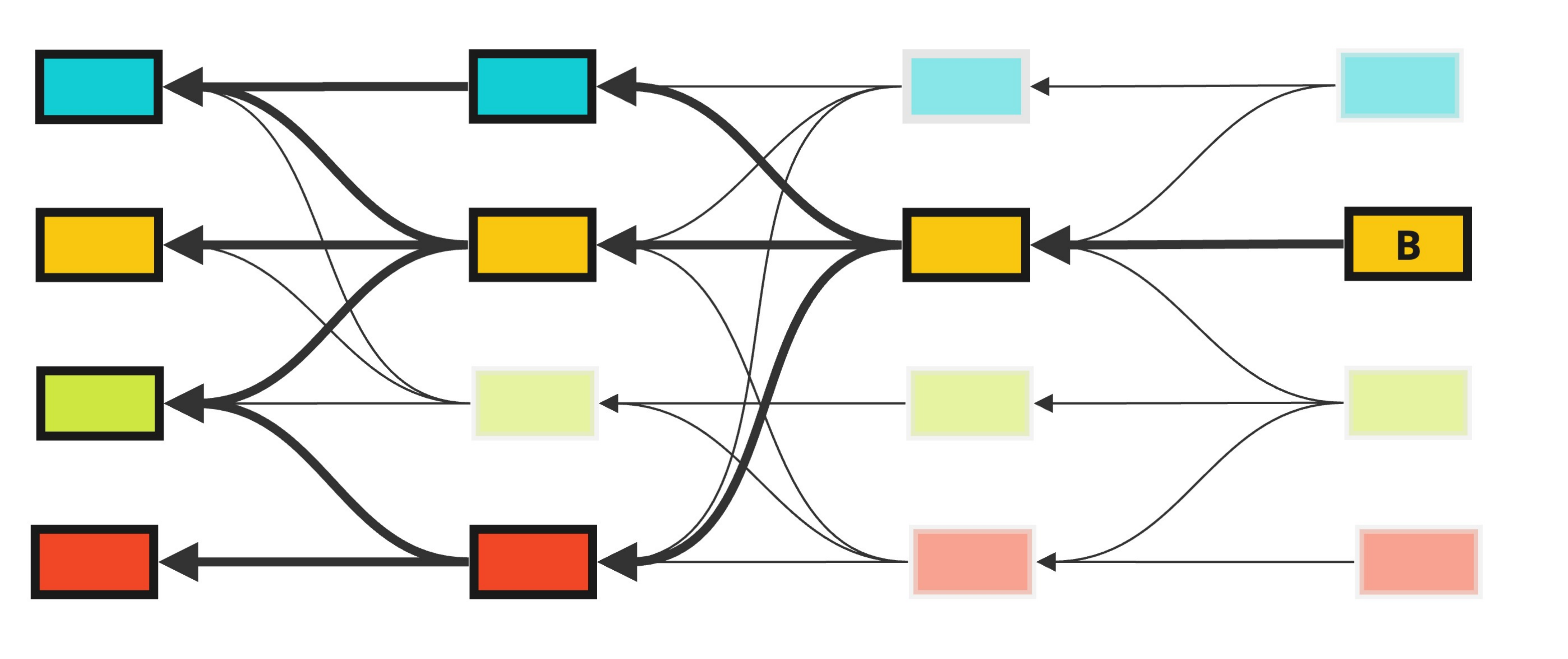}
\end{subfigure}%
\begin{subfigure}{.5\textwidth}
    \centering
    \includegraphics[trim={0 2cm 0 0},clip,width=.95\linewidth]{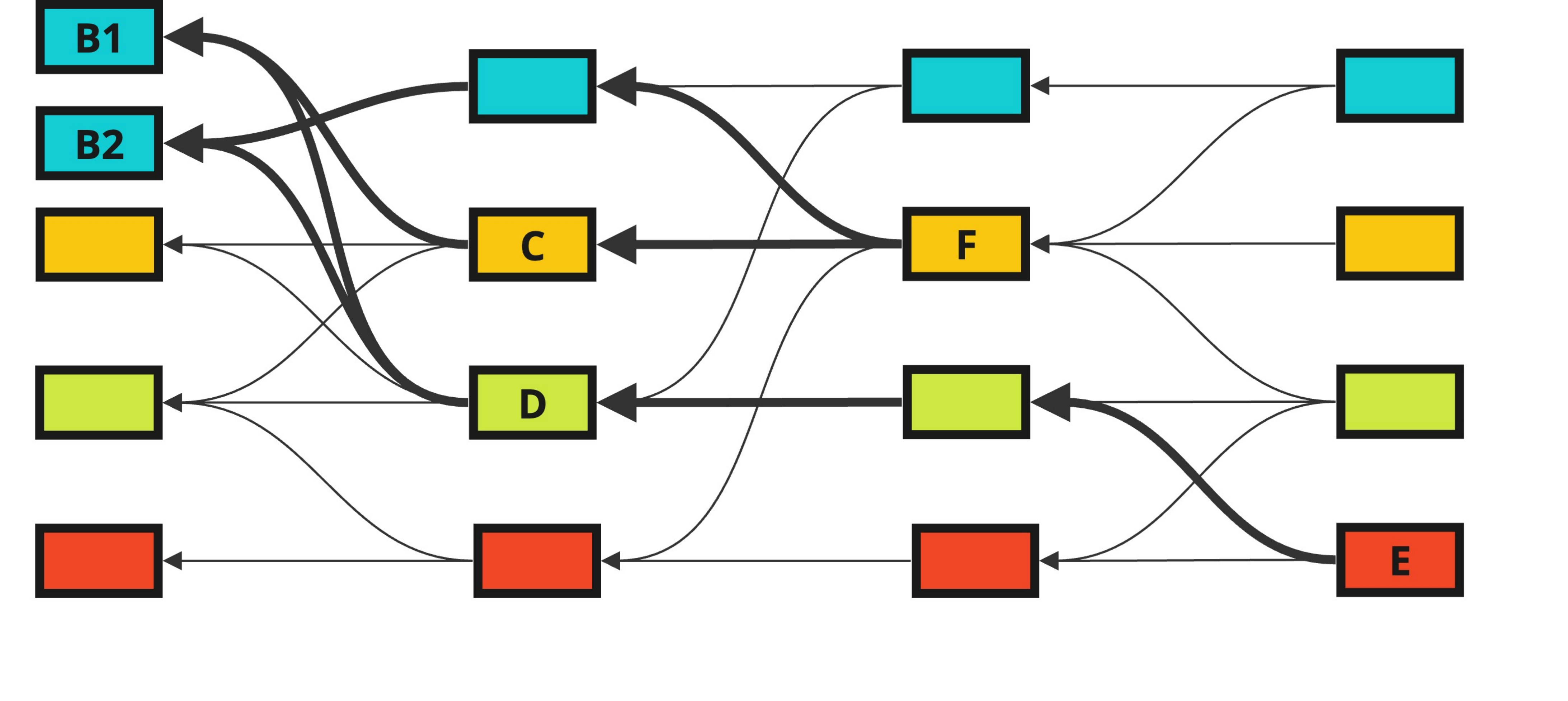}
\end{subfigure}
\caption{On the left, the past cone of a block $B$ consist of all blocks that are reachable from $B$. Blocks depicted with less opacity are not in the past cone, while blocks that appear `fatter' do form the past cone. 
On the right, $B1$ and $B2$ are equivocations of the blue node. The orange node issues a block $C$ referring only to $B1$ and not to $B2$. The green sees the equivocation immediately and refers with its block $D$ to both blocks $B1$ and $B2$ and adds the blue node to its set of equivocators. The orange node adds the blue node to the set of equivocators in the next instant with block $F$, and the red one finally witnesses the equivocation with its block $E$.
}
\label{fig: pastCone and Equivocation}
\end{figure}

Each node maintains locally its own DAG $\D$ and the set of all received blocks, denoted as $\buffer$.
In the protocol described later, a correct node issues a block every round and always references its previous block. Thus, all blocks of a correct node are linearly ordered in $\D$. 
\begin{definition}[Equivocation]
A node $\mu$ is an equivocator for node $\eta$ if two blocks created by $\mu$ are not linearly ordered in $\buffer$  maintained by $\eta$. 
\end{definition}
We assume that if a correct node detects an equivocation, then this node includes a proof 
in the next block. After receiving this block, every correct node checks the proof, identifies the equivocator, and updates its known sets of equivocators, denoted as~$\EqSet$; see Fig.~\ref{fig: pastCone and Equivocation}.


\begin{definition}[Quorum]
A set of blocks $S$ is called a quorum if $|\{\BS{B}{node}: B\in S\}|\ge 2f+1$.
\end{definition}

\subsection{Fast-path confirmation for UTXO transactions}
The next definition is motivated by cautious transactions, discussed in Sec.~\ref{sec: problem statements} and more formally in Def.~\ref{def: cautious account}.
\begin{definition}[Ready transaction]\label{def: valid transaction}
    For a given block $B$, a \text{UTXO} transaction $tx$ is called $B$-ready if the two conditions hold:  \textbf{\textup{1)}} the transaction $tx$ is included in $\BS{B}{txs}$, and  \textbf{\textup{2)}}  all transactions, whose outputs are the inputs of transaction $tx$, are confirmed in $\cone(B)$.
\end{definition}
\begin{definition}[Transaction approval by block]\label{def: transaction approval}
A block $C$ approves a \text{UTXO} transaction $tx$ in block $B$ if the three conditions hold: \textbf{\textup{1)}} transaction $tx$ is $B$-ready,  \textbf{\textup{2)}} block $B$ is reachable from $C$, i.e., $B\in \cone(C)$,  \textbf{\textup{3)}} no block, that contains a double-spend for $tx$, is reachable from $C$.
\end{definition}
Two instances of the same transaction $tx$, that is included in different blocks $B$ and $C$, i.e., $tx\in \BS{B}{txs}$ and $tx\in \BS{C}{txs}$, are not considered as a double-spend. However, the number of transaction approvals is computed independently for $tx$ in $B$ and $C$.
\begin{definition}[Transaction certificate]\label{def: transaction certificate}
A block $C$ is called a transaction certificate (\textup{TC}) for a \text{UTXO} transaction $tx$ in block $B$ if $\BS{C}{slot}\in \{\BS{B}{slot}, \BS{B}{slot}+1\}$ and $\cone(C)$ contains a quorum $S$ of blocks, each of which approves $tx$ in $B$.
\end{definition}

\begin{definition}[Fast-path confirmation]\label{def: confirmedTx}
A transaction $tx$ is said to be fast-path confirmed in a DAG if the DAG contains a block $B$ and a quorum of blocks, each of which is a certificate for $tx$ in $B$.
\end{definition}
We refer to Fig.~\ref{fig: txcertificate} for the concept of transaction certificates. By Def.~\ref{def: cautious account}, \ref{def: honest account}, and~\ref{def: transaction approval}, if the causal history of a block $C$ contains a quorum of blocks all observing an honest cautious transaction in a block $B$, then $C$ is a transaction certificate. Moreover, $tx$ is fast-path confirmed due to a quorum of TCs for $tx$ in $B$.

\begin{figure}[t!]
    \centering
\includegraphics[width=0.45\textwidth]{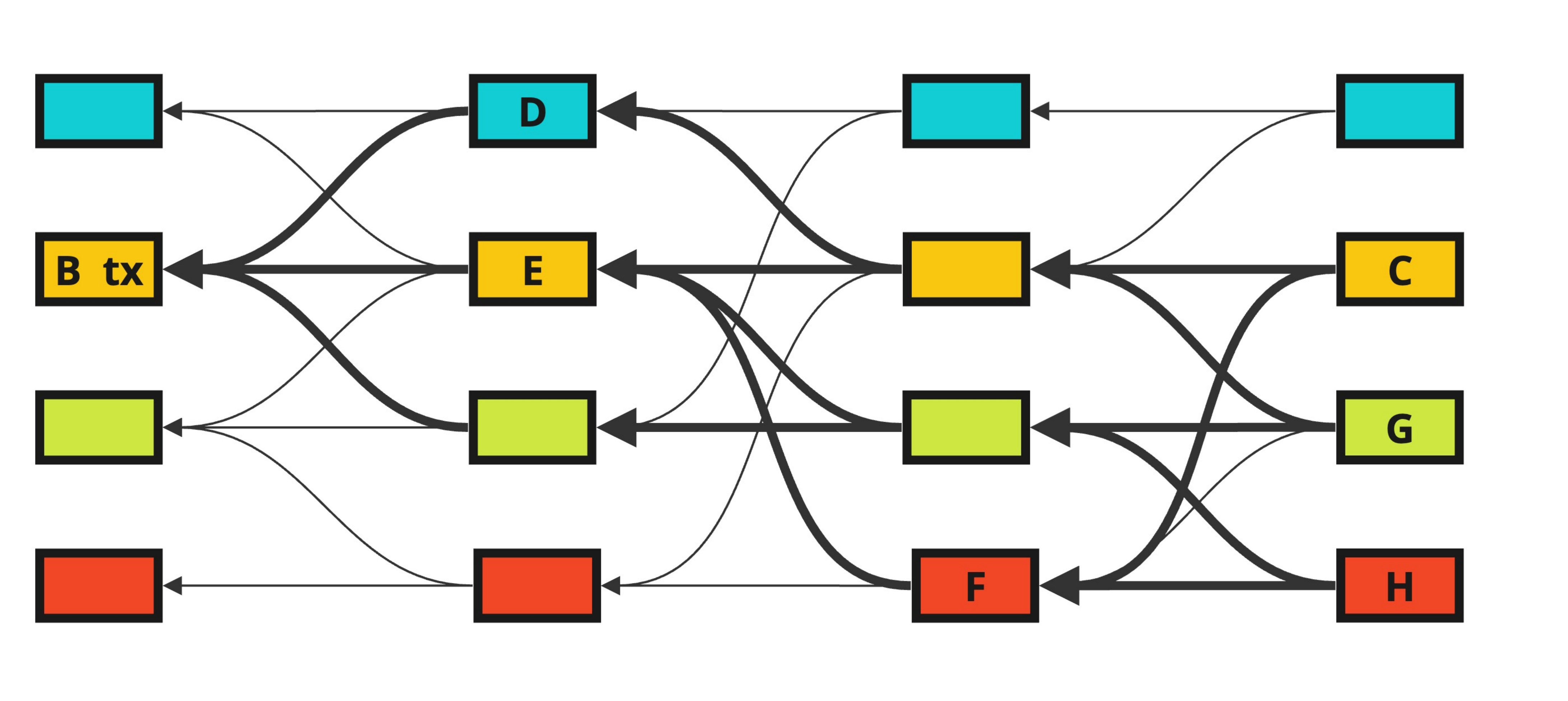}
    \caption{The block $B$ contains an honest cautious transaction $tx$, i.e., $tx \in \BS{B}{txs}$, and thus $tx$ is $B$-ready. The block $C$ issued by the orange node is a transaction certificate for $tx$ in $B$, as its past cone contains a quorum represented by blocks $D$, $E$, and $F$, approving block $tx$ in $B$. In the same way, blocks $G$ and $H$ form transaction certificates; hence, transaction $tx$ is fast-path confirmed. }
    \label{fig: txcertificate}
\end{figure}

\subsection{Slot digest}
We proceed with a definition of a slot digest that allows us to order blocks created before a certain slot optimistically. For a given slot digest $\sigma$, we write $\sigma.\mathrm{slot}$ to refer to the slot for which the digest is constructed. 

\begin{definition}[Slot digest]\label{def: commitment}
A digest $\sigma_s$ for a slot $s$ is defined based on a valid DAG at time round $\langle s+1,f+2\rangle$ in an iterative way: \begin{itemize}
    \item for $s\ge 1$, 
    $$
    \sigma_s=\Hash(\sigma_{s-1},\Concat(Blocks_{\le s})),
    $$
    where $Blocks_{\le s}$ denotes all blocks in the DAG that are created at or before slot $s$ and not committed by $\sigma_{s-1}$. Here, $\sigma_{s-1}$ is the digest included in all blocks in the DAG, which are created at slot $s+1$ and round $\le f+1$. A function $\textsc{Concat}(\cdot)$ orders blocks in a certain deterministic way and concatenates them. We say the digest $\sigma_s$ commits all the blocks committed by $\sigma_{s-1}$ and the blocks in set $Blocks_{\le s}$;
    \item $\sigma_{-1}=0$ and $\sigma_{0}=\Hash(0, Genesis)$.
\end{itemize} 
Furthermore, we write $\BeforeCommit(\sigma_s)$ to denote $\sigma_{s-1}$ and write $Blocks_{=s}$ to denote all blocks created at slot $s$. 
\end{definition}
We refer to Fig.~\ref{fig: commitments} to demonstrate how digests are included in blocks; in particular, for the same slot, nodes could generate different digests. See Appendix~\ref{sec: digest correct} for details on when slot digests are included in blocks in a valid DAG.
\begin{definition}[DAG associated with digest]\label{def: DAG associated}
For a given slot digest $\sigma$, we write $\D(\sigma)$ to denote the DAG whose vertex set is all blocks committed to $\sigma$. We write $\EqSet(\sigma)$ for the set of equivocators that can be found based
on $\D(\sigma)$.    
\end{definition}
By definition, a digest for slot $s$ corresponds to a unique sequence of committed blocks created at or before slot $s$. One can invoke $\Order(\sigma_s)$ to get the order of blocks in $\D(\sigma_s)$. Specifically, the order is defined iteratively:  
\begin{equation}\label{eq: ordering}
{\Order(\sigma_s)= \Order(\sigma_{s-1}) \,||\, \Concat(Blocks_{\le s})},
\end{equation}
where $Blocks_{\le s}$ are the blocks in $\D(\sigma_s)\setminus \D(\sigma_{s-1})$ with $\sigma_{s-1}=\BeforeCommit(\sigma_s)$, and the base order is $\Order(\sigma_0)=(Genesis)$. Each node maintains its optimistic order of blocks written as $\Order_{\text{own}}$. This order is defined as $\Order(\sigma)$, where $\sigma$ is the digest adopted by the node.

Slot digests form a slot digest tree, where the path between the root and any vertex in the tree determines a backbone chain.  If two correct nodes agree on digest $\sigma$, they share the same subDAG $\D(\sigma)$ in their respective DAGs. 
\begin{definition}[Backbone chain]
    For a given digest $\sigma_s$, define the backbone chain associated with $\sigma_s$ as 
    $$
    \chain(\sigma_s)=(\sigma_{-1},\sigma_{0},\ldots,\sigma_{s-1}, \sigma_s),
    $$
    where $\sigma_{j-1} = \BeforeCommit(\sigma_{j})$ for $j\in [s]$. When the backbone chain $\chain$ is clear from the context, we also write $\chain[j]$ to denote $\sigma_j$.
\end{definition}
\begin{definition}[Consistent and conflicting digest]\label{def: consistent and consistent}
A digest $\sigma$ is called consistent with a digest $\theta$ if the backbone chain $\chain(\sigma)$ contains $\theta$. Two digests $\sigma$ and $\theta$ are called conflicting if  neither $\sigma$ is consistent with $\theta$, nor $\theta$ is consistent with $\sigma$. To check whether two digests are conflicting, nodes use a function $\IsConflict(\theta, \sigma)$ that outputs $1$ in case they are conflicting and $0$ otherwise.
\end{definition}



\subsection{Finality for slot digests}
Similar to transaction certificates, we introduce the notion of digest certificates that allow nodes to finalize backbone chains and secure finality for a chain switching rule (see Sec.~\ref{sec: protocol}.\ref{sec: chain-switching rule}).
\begin{definition}[Digest certificate] \label{def: quorum certificate}
A digest certificate (\textup{DC}) for a digest $\sigma$ is a block $A$ with $\BS{A}{slot}=\BS{\sigma}{slot}+2$ such that $\cone(A)$ contains a quorum of blocks, all issued at the same slot $\BS{A}{slot}$ and include the same digest $\sigma$. Note that either $\BS{A}{digest}=\sigma$, or $\BeforeCommit(\BS{A}{digest})=\sigma$. If $A$ is a \textup{DC}, nodes can invoke a function $\Commit(A)$ to get the digest $\sigma$, certified by $A$.
\end{definition}

\begin{figure}[t]
    \centering
\includegraphics[width=\textwidth]{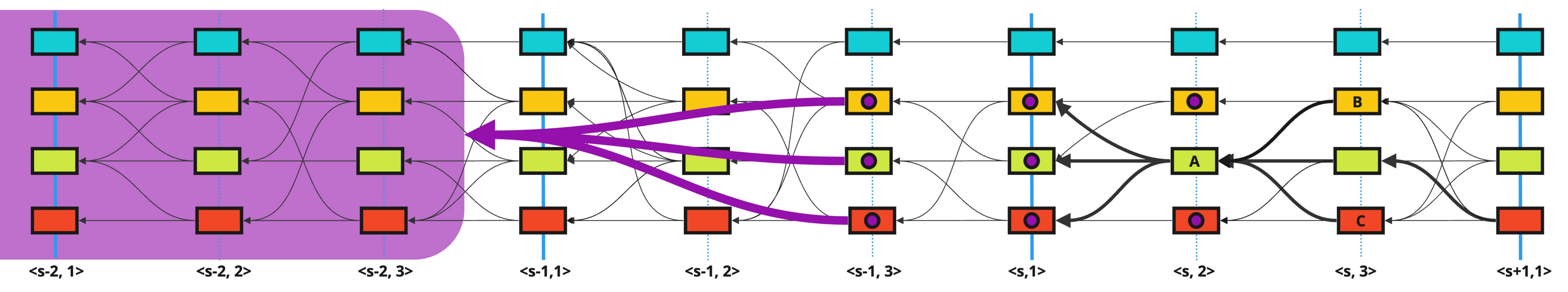}
    \caption{The blocks of the orange, green, and red nodes created at round $\langle s-1, 3\rangle$ contain a digest of the DAG up to slot $s-2$. Here, the fat purple edges are not normal references but represent the $\BS{B}{digest}$ entries, and this digest is depicted using a purple circle. The blue node has a different perception of what happened in slot $s-2$ and will create a different digest not shown in the image. Block $A$ does contain blocks of a quorum with the same digest in its past cone, and hence serves as a digest certificate (DC) of the ``purple'' digest. This digest becomes final as blocks $B$ and $C$ also form DCs. }
    \label{fig: commitments}
\end{figure}

\begin{definition}[Final digest]\label{def: final commitment}
A digest $\sigma$ is called final in a valid DAG if the DAG contains a quorum of \textup{DCs} for $\sigma$. In addition, if $\sigma$ is final, then all digests in the backbone chain ending at $\sigma$, i.e., $\chain(\sigma)$, are also final. 
\end{definition}
In Fig.~\ref{fig: commitments}, blocks $A$, $B$ and $C$ are certificates for the same digest and the digest gets final due to a quorum of certificates. 

Each node maintains a view of the final block order, denoted as $\Order_{\text{final}}$, defined by $\Order(\sigma)$, where $\sigma$ is the latest final digest.
While a digest that is final for one correct node will eventually become final for all other correct nodes, the finalization of the same digest might occur at different slots for different nodes. To have a more robust perception of the finality time of digests, we introduce the following definition.
\begin{definition}[Finality time]\label{def: final time}
Let $\D$ be a valid DAG. Let $s_{\textup{final}}$ be the slot index of the last final digest $\sigma_{s_{\textup{final}}}$ in $\D$. Let $s_{\textup{pre}}$ be the slot index of the last final digest in $\D(\sigma_{s_{\textup{final}}})$. Then for any digest $\sigma_s$ with $1\le s\le s_{\textup{pre}}$,  one can define the finality time $\FinalTime(\sigma_s)$ as the first slot index $\tau$ such that $\sigma_s$ is final in $\D(\sigma_{\tau})$\footnote{Note that the finality time of some digests may coincide.}. For $s>s_{\textup{pre}}$, we set $\FinalTime(\sigma_s)$ to be $\bot$.
\end{definition}
From the safety property of digest finality (see Appendix~\ref{sec: proof for finality for ELSS}), it follows that if $\FinalTime(\sigma_s)$ is defined not as $\bot$, then their values are the same for two correct nodes. While the notion of a finality time is defined for digests, one can define $\FinalTime(s)$ for a slot $s$ as $\FinalTime(\sigma)$ with $\BS{\sigma}{slot}=s$.

\subsection{Consensus-path confirmation for UTXO transactions}\label{sec: transaction finality}
Given a valid DAG $\D$ with a backbone chain $\chain$, the notion of finality time 
allows to partition the DAG. Specifically, let $(\tau_1,\dots,\tau_{i+1})$ be the vector of all distinct finality times for digests in $\chain$ (here, $\tau_1$ is assumed to be $0$); then one can partition the DAG $\D$ using  the corresponding digests: 
$\D(\sigma_{\tau_{j+1}})\setminus D(\sigma_{\tau_j})$ for $2\le j \le i$. Informally, the consensus-path confirmation is defined using this partition. First, all transactions with transaction certificates in the partition are attempted to be added to $\LOG$. Second, by traversing the blocks in the partition in the final order $\Order_{\text{final}}$, all transactions that don't have certificates are attempted to be added to $\LOG$. The attempt is successful if no conflicting transaction is already in $\LOG$ and the transactions creating transaction inputs are already in $\LOG$, i.e., when the UTXO ledger properties are not violated (see Sec.~\ref{sec: system model}). More formally, we refer for consensus-path confirmation to procedure $\FinalizeTransactions()$ in Alg.~\ref{alg: ledger}.




\section{Protocol description}\label{sec: protocol}
During the execution of the protocol, each node maintains its own DAG $\D$, which can only be extended by adding new vertices in the state update phase. During slot $s$, a node adapts one backbone chain $\chain(\sigma_{s-2})=(\sigma_{-1},\sigma_0,\ldots,\sigma_{s-2})$. At the last round of slot $s$, the node updates the chain with the digest $\sigma_{s-1}$ for slot $s-1$ such that $\BeforeCommit(\sigma_{s-1})=\sigma_{s-2}$. Switching the backbone chain can happen only at the state update phase of the first round of a slot. Switching the chains is important for synchronization after round $GST$ when the asynchronous period ends. The broadcasting of blocks during the first $f+1$ rounds of slot $s$ is used to ``synchronize'' the nodes' perception of the previous slot $s-1$.
In particular, after $GST$, all correct nodes with the same initial digest for slot $s-2$ do agree on the blocks of slot $s-1$ at round $\langle s, f+2\rangle$ and do create the same slot digest for slot $s-1$. The flow of procedures is described in Alg.~\ref{alg: consensus}. The following sections delve into the different components of the protocol.\\
\begin{algorithm}[t]
\scriptsize
\caption{Slipstream protocol}
\label{alg: consensus}
\Loc{}{
$\D\gets \{Genesis\}$ \tcp*{DAG maintained by the node}
$\buffer\gets \{Genesis\}$ \tcp*{Buffer maintained by the node}
$\chain\gets (0)$ \tcp*{backbone chain adapted by the node}

$\LOG\gets \{\}$ \tcp*{Ledger maintained by the node}
$\Order_{\text{own}} \gets(Genesis)$ \tcp*{Optimistic order of blocks}
$\Order_{\text{final}} \gets(Genesis)$ \tcp*{Final order of blocks}
$s_{\text{final}}\gets 0$\tcp*{Slot index of the last final slot digest}
$s_{\text{pre}}\gets 0$ \tcp*{Slot index of the last final slot digest in $\D(s_{\text{final}})$, see Def.~\ref{def: final time}}
$\text{pk}, \text{sk}$ \tcp*{Public and secret keys of the node}

$B_{\text{own}} \gets \{\}$ \tcp*{Last block created by the node}
$\EqSet\gets \{\}$ \tcp*{Set of equivocators known to the node}
}
\For{\textup{slot index} $s=1,2,3, \ldots$}{
\For{\textup{round index} $i=1,\ldots,f+2$}{
\textbf{\textsc{Receive phase}}\;
\textsc{ReceiveMessages()} \tcp*{see Sec.~\ref{sec: protocol}.\ref{sec: receiving messages}}
\hrulefill

\textbf{\textsc{State update phase}} \;
\textsc{UpdateHistory()} \tcp*{see Sec.~\ref{sec: protocol}.\ref{sec: update the history}}
\textsc{UpdateDAG}() \tcp*{see Alg.~\ref{alg: updateDAG}}

\If{$i=1$}{
\eIf{\textup{the node was slot}-$(s-1)$ \textup{asleep}}{$\WakeUpChain()$ \tcp*{see line~\ref{proc:WakeUpChain} in Alg.~\ref{alg: chain}}}{$\SwitchChain()$ \tcp*{see line~\ref{proc:SwitchChain} in Alg.~\ref{alg: chain}}}} 
\If{$i=f+2$}{$\UpdateChain()$ \tcp*{see line~\ref{proc:UpdateChain} in Alg.~\ref{alg: procedures CreateBlock UpdateChain FinalizeSlot}}}
$\ConfirmTransactions()$ \label{proc: fast path} \tcp*{see Alg.~\ref{alg: ledger}}
$\FinalizeSlots()$ \tcp*{see line~\ref{proc:FinalizeSlots} in Alg.~\ref{alg: procedures CreateBlock UpdateChain FinalizeSlot}}
$B_{\text{own}}\gets \CreateBlock{}$ \label{alg: bown created}\tcp*{see line~\ref{fn: create block} of Alg.~\ref{alg: procedures CreateBlock UpdateChain FinalizeSlot}}
\hrulefill

\textbf{\textsc{Send phase}} \;
$\Broadcast(B_{\text{own}})$ \tcp*{see Sec.~\ref{sec: protocol}.\ref{sec: best-effort broadcast} }
} }
\end{algorithm}
\begin{enumerate}
\item \label{sec: creating a block}\textbf{Creating a block:} Creating a new block $B_{\text{own}}$ with timestamp $\langle s,i \rangle$ happens at the end of the state update phase, see line~\ref{alg: bown created} of Alg.~\ref{alg: consensus}. For this purpose, nodes use function $\CreateBlock()$, see Alg.~\ref{alg: procedures CreateBlock UpdateChain FinalizeSlot}.  The node includes the timestamp $\langle s,i\rangle$, a set of transactions, the currently adopted slot digest $\sigma_{s-2}$,\footnote{This happens for round $\langle s,i\rangle$ with $i<f+2$. For round  $\langle s,f+2\rangle$, the node includes $\sigma_{s-1}$, the digest corresponding to slot $s-1$.} and hash references to all unreferenced blocks in its DAG.\\ 
\item \label{sec: best-effort broadcast} \textbf{Best-effort broadcast:} The key principle in broadcasting a newly created block $B_{\text{own}}$ using $\textsc{Broadcast()}$ is to ensure that the recipient knows the past cone of the new block and, thus, might be able to update its DAG. Each node needs to track the history of the known blocks by every other node, e.g., in the perception of a given node, the set of known blocks by node $\eta$ is denoted as $\History[\eta]$. So, when sending the newly created block $B_{\text{own}}$, the node sends to node $\eta$ the set $\cone(B_{\text{own}})\setminus \History[\eta]$ and makes the update as follows $\History[\eta]\gets \History[\eta]\cup \cone(B_{\text{own}})$.\\ 
\item \label{sec: receiving messages} \textbf{Receiving messages:} 
In the receive phase of a round, the protocol executes a procedure called $\textsc{ReceiveMessages}()$, which allows to receive messages, i.e., blocks from nodes. When receiving a new block $B$, the node puts $B$ into $\buffer$. Before $GST$, the node might not receive a block $B$ with $\BS{B}{time} < GST$ created by a correct node. After $GST$, each block $B$ with timestamp $\langle s, i\rangle= \BS{B}{time} \ge GST$, created by a correct node, is received by any other correct node at the next round $\langle s, i+1\rangle$ (or $\langle s+1, 1\rangle$ in case $i=f+2$). \\ 
\item \label{sec: update the history} \textbf{Updating the history:}
The state update phase includes the procedure $\textsc{UpdateHistory()}$, which updates the set of known blocks by other nodes. When the causal history of a block $B$ is known to the node, i.e., $\cone(B)\subset \buffer$, the node updates the history $\History[\eta]\gets \History[\eta]\cup \cone(B)$, where $\eta=\BS{B}{node}$.\\
\item \textbf{Updating the backbone chain:}
The state update phase of the last round $\langle s, f+2\rangle$ of slot $s$ includes the procedure, $\UpdateChain()$, that extends the backbone chain $\chain(\sigma_{s-2})$ by digest $\sigma_{s-1}$, see Def.~\ref{def: commitment} and Alg.~\ref{alg: procedures CreateBlock UpdateChain FinalizeSlot}. It also naturally orders blocks in the maintained DAG that are created at or before slot $s-1$ and updates $\Order_{\text{own}}$ as $\Order(\sigma_{s-1})$.\\
\item \textbf{Updating the DAG:}
The procedure $\UpdateDAG()$ in Alg.~\ref{alg: updateDAG} updates the DAG $\D$ maintained by the node. Suppose the current timestamp is $\langle s, i\rangle$ and that the buffer contains a block $B$ with timestamp $\BS{B}{time}=\langle s, i-1\rangle$ (or $\langle s-1, f+2\rangle$ for $i=1$) and the same digest as the node adopts. If the past cone $\cone(B)$ is in $\buffer$; the node adds to the DAG $\D$ the whole past cone $\cone(B)$ if all the updating (U) conditions are met: 

\begin{enumerate}
 \item[U1]  No block at slot $s$ in $\cone(B)$ is created by an equivocator in $\EqSet(\sigma_{s-2})$, see Def.~\ref{def: DAG associated} and line~\ref{alg: check equivocators} of Alg.~\ref{alg: updateDAG}. \label{item: U1}
 \item[U2] $\cone(B)$ is valid, e.g., all blocks in $\cone(B)$ that are created at slot $s$ contain the same digest $\sigma_{s-2}$, see Def.~\ref{def: validDAG} and line~\ref{alg: add cone} of Alg.~\ref{alg: updateDAG}.\label{item: U2}
 \item[U3] For every block $C$ in $\cone(B)\setminus \D$ which is not committed to the chain $(\sigma_{-1},\sigma_0,\ldots,\sigma_{s-2})$ and created at slot $s-1$ or before, the node computes its \textit{reachable number} $\ReachNumber$ $(C,B)$ - the number of distinct nodes who created blocks in $\cone(B)$ at slot $s$ from which one can reach $C$. It must hold that $\ReachNumber(C,B)\ge i-1$, see line~\ref{alg: reachnumber} of Alg.~\ref{alg: updateDAG}. \label{item: U3}
\end{enumerate} 
\item \label{sec: chain-switching rule}\textbf{Chain switching rule:} 
If the node was slot-$s$ awake, the node calls the procedure $\SwitchChain()$ at round $\langle s+1, 1 \rangle$, see Alg.~\ref{alg: chain}.
\begin{figure}[t]
\centering
\includegraphics[width=\textwidth]{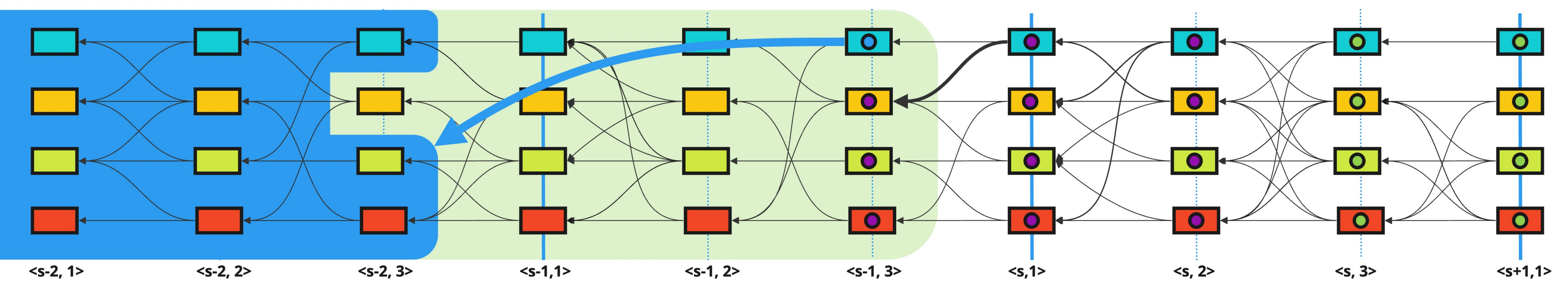}
\caption{We are in the situation, which is similar to Fig.~\ref{fig: commitments} up to round $\langle s,1\rangle$. Now, the blue node realizes that its digest on the ``blue'' past is not the one of the majority and merges its DAG with the past cone of the leader, which is the orange node. So, it switches to the ``purple'' chain at round $\langle s, 1 \rangle$. As before, the purple digest becomes final, and at round $\langle s, 3 \rangle$, all nodes have the same perception of what happened until slot $s-1$ and create the same ``green'' digest.  }
\label{fig: chainswitching}
\end{figure}

First, the node checks how many other nodes generated the same digest $\sigma_{s-1}$ as the node after running $\UpdateChain()$ at the previous round.
Let $N_\text{total}$ be the number of nodes having issued a block at round $\langle s, f+2 \rangle$ that is received by a node at round $\langle s+1, 1 \rangle.$ Let $N_\text{same}$ be the number of those nodes having the same digest as the given node.  To compute these values, the node uses the function $\LastBlockFrom(node)$ that outputs the latest block created by $node$ and stored in  $\buffer$. 

Before checking the condition for switching the chain, the node runs the procedure $\CheckELSS()$ in Alg.~\ref{alg: chain} that updates the indicator $I_{\text{ELSS}}$ to $1$ if there have been two different digests for slot $s-1$ with more than $f+1$ supporters. Recall that we have two different communication models, the ELLS model, Def.~\ref{def:ELSS}, and the SS model, Def.~\ref{def: slot-sleepy model}.  Having two different digests with more than $f+1$ supporters indicates that we are not in the SS model. Similarly, in line \ref{alg: change ELSS 1}, the digest certificate of any node can not conflict with the digest of a correct node in the SS model, and if so, $I_{\text{ELSS}}$ turns to $1$. 


If the node has not finalized the digest at the previous slot, $s$, a leader-based tie-breaker mechanism is employed.  To determine the leader, the node invokes the $\Leader()$ function, Def.~\ref{def: leader function}. 
This procedure allows the node to switch its chain, see line~\ref{alg: has not finalized and leader}, in two following cases:
\begin{enumerate}
    \item $2 N_\text{same}\leq N_\text{total}$: if the digest of the leader is consistent with the own digest OR the node's digest certificate is created not later than the one of the leader, the node updates the DAG by adding the past cone of the leader block $B_\text{leader}$ and adopts the leader's backbone chain together with its induced optimistic ordering for blocks. 
    \item $2 N_\text{same}> N_\text{total}$: if the ELSS indicator is true, i.e., indicating that the underlying model is the ELSS model, AND the node's digest certificate is created not later than the one of the leader, the node merges the past cone of the leader block with its maintained DAG and adopts the leader's backbone chain and induced final ordering for blocks.  
\end{enumerate}
\item \label{sec: waking up} \textbf{Waking up rule:}
Suppose a node is slot-$(s+1)$ awake and was slot-$s$ asleep. The node can detect this by inspecting the local DAG at round $\langle s+1,1 \rangle$ and checking the existence of block $B_{\text{own}}$ at round $\langle s,f+2 \rangle$. If no such block exists, then the node proceeds with $\WakeUpChain()$, see line \ref{proc:WakeUpChain} in Alg.~\ref{alg: chain}.

Using function $\Mode(M)$, the node checks which digest, say $\sigma_{s-1}$, is the most frequent one among all created by non-equivocator nodes at round $\langle s, f+2\rangle$. Then the node adopts the backbone chain to the chain corresponding to $\sigma_{s-1}$ and adapts the induced ordering for blocks. In addition, it updates the local DAG with past cones of all blocks that include $\sigma_{s-1}$.\\
\item \label{sec: finality}\textbf{Finality rule:}  
Using Def.~\ref{def: final commitment}, the node finalizes slot digests using procedure $\FinalizeSlots()$ at line \ref{proc:FinalizeSlots} in Alg.~\ref{alg: procedures CreateBlock UpdateChain FinalizeSlot}. For this purpose, the node tries to find a quorum of digest certificates.
Once a digest $\sigma$ is final, a final ordering $\Order_{\text{final}}$ of the blocks committed to $\sigma$ is obtained using the function $\Order()$, see Eq.~\eqref{eq: ordering}. The exact nature of this order is not important for the results of this paper as long as it preserves the causal order of the blocks.\\
\item \label{sec: confirmation rule}\textbf{Confirmation rule:} UTXO transactions can be confirmed either through the fast path or the consensus path. At state update of each round the protocol runs 
$\ConfirmTransactions()$ for fast-path confirmation, see line~\ref{proc: fast path} in Alg.~\ref{alg: consensus}. The DAG $\D$ is inspected to find a quorum of transaction certificates for all possible transactions; see Def.~\ref{def: confirmedTx} and line~\ref{proc: confirmTransactions} in Alg.~\ref{alg: ledger}. For the consensus path, the protocol runs procedure $\FinalizeTransactions()$. As said in Sec.~\ref{sec: transaction finality}, the consensus path relies on the notion of the finality time. Since the latest finality time could be updated only when the last final slot is updated, $\FinalizeTransactions()$ is executed inside $\FinalizeSlots()$ at line~\ref{proc: finalizeTransactions} in Alg.~\ref{alg: procedures CreateBlock UpdateChain FinalizeSlot}. In procedure $\FinalizeTransactions()$, one first gets the latest finality time $\tau$. Then, for each transaction included in blocks $B\in D(\sigma_{\tau})$ with $B.\text{slot} \le \tau-2$, one checks the existence of a transaction certificate, and if such one exists, the transaction is attempted to be added to $\LOG$. The attempt is successful if the two UTXO ledger properties are preserved (see Sec.~\ref{sec: system model}). Then, using the total order, all transactions in blocks $B\in D(\sigma_{\tau-2})$ without certificates are similarly attempted to be added to $\LOG$. 
\end{enumerate}

\section{Formal results} \label{sec: results}
\begin{theorem}[Safety and liveness for $\Order_{\text{own}}$]\label{th: dynamically available}
Assume the network operates in the \textup{SS} model, e.g., $GST=0$. Assume a majority of slot-$s$ awake correct nodes for all slots $s\in\mathbb{N}$, i.e., the number of correct nodes is strictly larger than the number of Byzantine nodes for every slot. Then it holds
\\ \textbf{Safety:} Every slot-$s$ awake correct node adopts the same digest for slot $s-2$, i.e., for any two sequences of blocks maintained by any two correct nodes in $\Order_{\textup{own}}$, one sequence must be a prefix of the other.
\\\textbf{Liveness:} Every block created by a correct node at slot $s$ is included in the digest of slot $s$, i.e., it appears in $\Order_{\textup{own}}$ after slot $s+1$.
\end{theorem}

\begin{theorem}[Safety and liveness for $\Order_{\text{final}}$]\label{th: eventually syncrony}
Assume the network operates in the \textup{ELSS} model (e.g., $GST>0$ is unknown to nodes).
Assume a supermajority of correct nodes, i.e., $n=3f+1$. Then it holds
\\ \textbf{Safety:}  
For any two sequences of final blocks maintained by any two correct nodes in $\Order_{\textup{final}}$, one sequence must be a prefix of the other.
\\ \textbf{Liveness:} Before $GST$, if a correct node creates a block, then the block eventually appears in $\Order_{\textup{final}}$. After $GST$, there is a slot $s_{\textup{same}}$ after which every block created by a correct node at slot $s\ge s_{\textup{same}}$ appears in $\Order_{\textup{final}}$ after round $\langle s+2, 2\rangle$. The number of slots $s_{\textup{same}}-GST$ is stochastically dominated by  $\xi_1+\xi_2$, where $\xi_1$ and $\xi_2$ are independent and identically distributed geometric random variables with parameter $p=\tfrac{1}{n}$. In particular, the expected number of slots to sync correct nodes with the same slot digest satisfies $\mathbb{E}(s_{\textup{same}})\le GST+2n$.
\end{theorem}
\begin{theorem}[Safety, liveness, and consistency for $\LOG$]\label{th: confirmation}
Assume the network operates in an \textup{ELSS} model.
Assume a supermajority of correct nodes, i.e., $n=3f+1$. Then it holds
\\ \textbf{Safety:} Two UTXO transactions that constitute a double-spend never both appear in $\LOG$. 
\\\textbf{Liveness:} Every cautious honest UTXO transaction  eventually added to $\LOG$ through the consensus path. There is a round after $GST$ such that every cautious, honest UTXO transaction, included in a block by a correct node at round $i$, gets fast-path confirmed and appears in $\LOG$ after round $i+2$.
\\ \textbf{Consistency:} If a UTXO transaction is fast-path confirmed by a correct node, then it eventually appears in $\LOG$ of every correct node.
\end{theorem}
\begin{remark}\label{rem: on communication complexity}
 The amortized communication complexity (for one transaction and big block size) in Slipstream in the ELSS model is $O(n^2)$, which is similar to all other DAG-based BFT protocols that don't use special erasure-code-based BRB primitives such as~\cite{cachin2005asynchronous,das2021asynchronous}. Indeed, because of potentially $n$ hash references to previous blocks, the block size in our protocol is (at least) linear with $n$, the number of nodes. Thereby, one can put $n$ transactions in a block without affecting the order of the block size. Each block needs to be sent by every correct node to every other correct node, i.e., each block makes $O(n^2)$ trips between nodes. After $GST$, consider all transactions in $s$ consecutive slots. For confirming the $s (f+2) n(n-f)$ transactions, committed to these $s$ slots by $n-f$ correct nodes, it takes $O((s+1)(f+2)n^4)$ bits to be communicated, i.e., $O(n^2)$ bits per transaction bit.
\end{remark}

\section{Proofs}\label{sec: proofs}

\subsection{Proof of Theorem~\ref{th: dynamically available}}
The key ingredient for Theorem~\ref{th: dynamically available} is the following key lemma, which can be applied to both the ELSS model and the SS model. It states that under synchrony, two correct nodes starting the slot with the same slot digest will generate identical digests at the end of the slot. The proof of this lemma is inspired by~\cite{dolev1983authenticated}, and it makes use of the lock-step nature of the network models considered in the paper.
\begin{lemma}\label{lem: same sc}
Suppose slot $s$ occurs after $GST$. If two correct nodes have adopted the same digest for slot $s-2$ after completing the state update phase of round $\langle s,1\rangle$, then they generate the same digest for slot $s-1$ after completing the state update phase of round $\langle s,f+2\rangle$.
\end{lemma}
\begin{proof}
Let $\eta$ and $\mu$ be two correct nodes that adopted the slot digest $\sigma_{s-2}$ after the state update phase of round $\langle s,1\rangle$. Consider an arbitrary block $B$ which is created at slot $s-1$ or before and included by node $\eta$ in the digest $\sigma_{s-1}$ for slot $s-1$. Toward a contradiction, assume that $B$ is not included by $\mu$ in the digest for slot $s-1$. Let $i$ be the first round of slot $s$ during the state update phase, of which $\eta$ has added $B$ to its maintained DAG. Note that $B$ was added together with the past cone of a certain block, say block $E$ (see line~\ref{alg: add cone} of Alg.~\ref{alg: updateDAG}). Technically, that could happen at slot $s-1$, and we treat $i$ as $1$ in this case. The block $B$ satisfies the condition U3 in the perception of node $\eta$. Thus, it holds that $\ReachNumber(B,E)\ge i - 1$.  Next, consider two cases depending on the value of $i$.

\textit{First case:} Consider $1\le i \le f+1$. Then, at the send phase of round $\langle s, i\rangle $, node $\eta$ creates block $C$ such that $\cone(C)$ contains $E$ (and consequently $B$) and sends it. At the next round $\langle s, i+1\rangle $, node $\mu$ receives the block $C$ with its past cone and updates the DAG with $B$, which leads to the contradiction.  Indeed, the buffer of $\mu$ contains $\cone(C)$ and all the U-conditions (see Sec.~\ref{sec: protocol}.6 are met:\\
U1: No block at slot $s$ in $\cone(C)$ is created by equivocator from $\EqSet(\sigma_{s-2})$ since $\BS{C}{node} = \eta$ and $\eta$ is a correct node that checked the same condition in all the previous rounds of slot $s$.\\
U2: The past cone of block $C$ is valid due to Lem.~\ref{lem:DAGisValid}.\\
U3: The reachable number $\ReachNumber(B,C)\ge i$, as the node $\eta$ was not counted previously. Similarly, one can show that the reachable number of all other blocks in $\cone(C)$, which are not yet included in the DAG of $\mu$, satisfies the same condition U3.

\textit{Second case:} Consider $i = f+2$. In this case, at the time of adding $B$ by $\eta$ to its maintained DAG, the reachable number $\ReachNumber(B,E)\ge i - 1 = f+1$. This means that $B$ was added to the DAG by at least one correct node at or before round $\langle s, f+1 \rangle $. By applying the arguments from the first case, we conclude that $\mu$ has included $B$ in its maintained DAG at or before round $\langle s, f+2 \rangle$. This leads to a contradiction.

The above arguments imply that after the state update phase of the round $\langle s, f+2\rangle$ both nodes $\eta$ and $\mu$ have in their DAGs the same set of blocks that are created at slot $s-1$ or before. Thereby, they generate the identical digest $\sigma_{s-1}$.
\end{proof}
\begin{proof}[Proof of Theorem~\ref{th: dynamically available}] First, we prove the safety property.
Let $\eta$ be a slot-$(s+1)$ awake correct node. We proceed by an induction on $s$. The base case $s=0$ holds as the digest $\sigma_{-1}$ is fixed and known to all nodes. In the following, consider $s\ge 1$. Consider two cases depending on whether node $\eta$ was awake or asleep at the previous slot. 

\textit{First case:} Suppose $\eta$ was slot-$s$ awake and adopted digest $\sigma_{s-2}$ in slot $s$. By Lem.~\ref{lem: same sc}, $\eta$ generated the same digest for slot $s-1$ at the end of slot $s$ as all other awake correct nodes who adopted $\sigma_{s-2}$ at the beginning of slot $s$. By the inductive hypothesis, the number of such correct nodes holds the majority among all slot-$s$ awake nodes, i.e., $2N_{\text{same}}\ge N_{\text{total}}+1$. By the inductive hypothesis, the indicator $I_{\text{ELSS}}$ will never change its value from $0$ in the SS model. Indeed, this indicator can be changed in lines~\ref{alg: change ELSS 1} and~\ref{alg: change ELSS 2} of Alg.~\ref{alg: chain}. As for line~\ref{alg: change ELSS 1}, it is not possible in the SS model to construct a digest certificate conflicting with the current digest of a correct node, i.e., $\IsConflict(\Commit(\text{DC}_{\text{leader}}), \BS{B_{\text{own}}}{digest})=0$. As for line~\ref{alg: change ELSS 2}, it is not possible to have two groups of blocks each of size at least $f+1$ including different digests at the end of any slot because all awake correct nodes always generate the same digests. Thus, when node $\eta$ will call $\SwitchChain()$ (see Alg.~\ref{alg: chain}), it will not pass the if-condition (see line~\ref{alg: adopt CC: SS safety}) as $I_{\text{ELSS}}=0$. Consequently, it will not switch the backbone chain after the state update phase of round $\langle s+1,1\rangle$.

\textit{Second case:} If $\eta$ was slot-$s$ asleep, then at the beginning of slot $s+1$,  $\eta$ calls the procedure $\WakeUpChain()$ (see Alg.~\ref{alg: chain}) and adopts the same digest as the majority of nodes generated at the end of slot $s$. 

In both cases, $\eta$ adopts the same digest as the majority of nodes produced at the end of slot $s$. This completes the proof of the inductive step and the proof of the safety property.

The liveness property follows directly from the synchronous nature of the SS model and the fact that every block sent by a correct node is received and added to a local DAG by any other correct node at the next round. 
\end{proof} 
\subsection{Proof of Theorem~\ref{th: eventually syncrony}} \label{sec: proof for finality for ELSS}
\subsubsection{Safety}
We start with a simple observation concerning digest certificates. We will make use of the property that $2f+1$ out of $3f+1$ nodes are correct.
\begin{lemma}\label{lem: qc are same}
  For any two digest certificates $A$ and $B$ with $\BS{A}{slot}=\BS{B}{slot}$, the digests certified by $A$ and $B$ are the same.
\end{lemma}
\begin{proof}
Let $s=\BS{A}{slot}=\BS{B}{slot}$. By Def.~\ref{def: quorum certificate}, a quorum of blocks created at rounds $\langle s, 1 \rangle, \ldots,\langle s, f+1 \rangle $ is needed for a digest certificate. Thus, at least one correct node created a block (or blocks) at slot $s$ that is (are) included in both $\cone(A)$ and $\cone(B)$. This implies that $A$ and $B$ certify the same digest.
\end{proof}

\begin{lemma}\label{lem: finalized locked}
If a correct node has finalized a slot digest $\sigma$ at slot $s$, then at least $f+1$ correct nodes have created blocks at slot $s$ serving as digest certificates for  $\sigma$.
\end{lemma}
\begin{proof}
    The statement follows directly from Def.~\ref{def: final commitment}. Indeed, the finalization of digest $\sigma$ could happen only when the correct node has in its DAG a quorum, i.e., $2f+1$, of blocks such that each block from this quorum is a DC for $\sigma$. Among these $2f+1$ blocks, $f+1$ were created by correct nodes.
\end{proof}

\begin{lemma}[Safety of finality]
Suppose a correct node has finalized a slot digest $\sigma$. Then no other correct node finalizes a slot digest conflicting with $\sigma$.
\end{lemma}
\begin{proof}
Toward a contradiction, assume that a conflicting digest $\theta$ is final for another correct node. Without loss of generality, assume $\sigma$ gets final at an earlier round, say slot $s$, by a given correct node.  By Lem.~\ref{lem: finalized locked}, $f+1$ correct nodes created digest certificates for $\sigma$ at slot $s$. By Lem.~\ref{lem: qc are same}, all digest certificates at slot $s$ certify the same digest. By the chain switching rule (procedure $\SwitchChain{}$ in Alg.~\ref{alg: chain}), no node of the $f+1$ correct nodes, created DCs, could switch its backbone chain to a one conflicting with $\sigma$ unless it received a digest certificate issued at a slot greater than $s$.

Let $B$ be a digest certificate created at a slot $r$ such that $\Commit(B)$ is conflicting with $\sigma$ and the slot $r>s$ is minimal among all such digest certificates. Such a certificate should exist since the conflicting digest $\theta$ gets final at some slot.  At least one of the said $f+1$ correct nodes had to contribute a block to form a digest certificate $B$. This means that this correct node has switched its chain at or before slot $r$ to the one conflicting with $\sigma$. But this could happen only if there would exist a digest certificate that conflicts with $\sigma$ and is issued before slot $r$. This contradicts the minimality of $r$.
\end{proof}

\subsubsection{Liveness}
\begin{lemma}\label{lem: tangles}
Suppose after round $GST$, all correct nodes adopt the same slot digest $\sigma_{s-2}$ for slot $s$. 
Then all correct nodes finalize the digest $\sigma_{s-2}$ at the state update phase of round $\langle s, 3\rangle$.
\end{lemma}
\begin{proof}
Every block created by one of the correct nodes at round $i$ is added to the DAG maintained by any other correct node at round $i+1$  by the arguments of Lem.~\ref{lem: same sc}. 
Correct nodes adopt the digest $\sigma_{s-2}$ for the whole slot $s$ and reference all previous blocks by the correct nodes. Hence, the $2f+1$ blocks created by correct nodes at round $\langle s, 2\rangle$ form a quorum of digest certificates for $\sigma_{s-2}$. By Def.~\ref{def: final commitment}, all correct nodes will finalize $\sigma_{s-2}$ at the state update phase of round $\langle s, 3\rangle$.
\end{proof}

Now we proceed with a statement showing when all correct node switch to the same backbone chain after $GST$.
\begin{lemma}[Liveness of finality]\label{lem: same scc}
After $GST$, let $s_{\textup{same}}$ denote the first slot when all correct nodes adopt the same slot digest. Then $s_{\textup{same}}-GST$ is stochastically dominated by a random variable which is a sum of two independent random variables having geometric distribution $\mathrm{Geo}(1/n)$. In particular, the expected number of slots it takes until all correct nodes follow the same backbone chain satisfies $\mathbb{E}(s_{\textup{same}})\le GST+2n$.
\end{lemma}
\begin{proof}
Let $s+1$ be an arbitrary slot after $GST$ and before moment $s_{\text{same}}$. 
Let $\mu$ be a correct node with a digest certificate generated at the highest slot number among all correct nodes.
Note that after $GST$, the function $\Leader()$ returns one of the nodes (not $\bot$); see Def.~\ref{def: leader function}. With probability of least $1/n$, node $\mu$ will be chosen as a leader based on the value of a common random coin for slot $s+1$. For every other correct node $\eta$, it holds DC$_{\text{leader}}.\mathrm{slot}\ge \text{DC}_{\text{own}}.\mathrm{slot}$. 
Let us check all potential scenarios when both $\eta$ and $\mu$ will not end up on the same backbone chain after running at $\langle s+1, 1 \rangle$ procedure $\SwitchChain()$, see Alg.~\ref{alg: chain}:
\begin{enumerate}
    \item node $\eta$ has finalized a digest at the previous slot, i.e., $s_{\text{final}}=s-2$ (see line~\ref{alg: finalized before} in Alg.~\ref{alg: chain}); 
    \item node $\eta$ has $s_{\text{final}}\neq s-2$ and $2N_{\text{same}} \ge N_{\text{total}}+1$ (see line~\ref{alg: adopt CC: SS safety} in Alg.~\ref{alg: chain}).
\end{enumerate}
In the first case, $\eta$ and $\mu$ have digest certificates created at the same slot $s$ and certifying the same digest $\sigma_{s-2}$ by Lem.~\ref{lem: qc are same}. In addition, they generated the same digests for slot $s-1$ by Lem.~\ref{lem: same sc}, i.e., $\eta$ will adopt the same digest for slot $s+1$ as $\mu$. 

In the second case, after $GST$, there could be some \textit{majority} groups of correct nodes that perceive $2N_{\text{same}}\ge N_{\text{total}}+1$, where correct nodes within each group have the same digest. Note that in this case, $N_{\text{total}}\ge 2f+1$ since blocks from all correct nodes will be received after $GST$. Therefore, $N_{\text{same}}$ has to be at least $f+1$. If there are two such majority groups, then all correct nodes will learn about that at the slot $s+1$ and change the indicator $I_{\text{ELSS}}$ to $1$ when running $\SwitchChain()$ at the next slot $s+2$ (see line~\ref{alg: checkElss}). Thus, we can assume that there is only one majority group of correct nodes (including $\eta$) that perceive $2N_{\text{same}} \ge N_{\text{total}}+1$ and $\eta$ has $I_{\text{ELSS}}=0$ (the latter condition does not allow $\eta$ to switch its chain to the one of the leader $\mu$). However, this also means that digests $\Commit(\text{DC}_{\text{leader}})$ and $\BS{B_{\text{own}}}{digest}$ are not conflicting (see line~\ref{alg: change ELSS 1}). In such a case, we note that with probability $1/n$, node $\eta$ could be selected as a leader for slot $s+1$ and then $\mu$ (as all other correct nodes outside the majority group) would switch its backbone chain as all required if-conditions in line~\ref{alg: sneak to non-conflic} are satisfied.

It remains to compute the expected number of slots after $GST$, which is sufficient for all correct nodes to switch to one chain. With a probability of at least $1/n$, a proper correct node will be selected as a leader, resulting in switching all correct nodes to one chain. However, the adversary controlling the Byzantine nodes has a one-time opportunity to create at least two majority groups that will not allow all correct nodes to adopt one chain. In the worst case, the adversary could use this one-time option only when a proper correct leader is chosen. In this case, all correct nodes set their indicators $I_{\text{ELSS}}$ to $1$ at the end of the slot and switch their backbone chains to the one of a proper correct leader next time. Thereby, the moment $s_{\text{same}}-GST$ is stochastically dominated by a sum of two independent random variables having a geometric distribution with probability $1/n$.
\end{proof}

\begin{theorem}[Liveness of finality]\label{th: liveness of finality}
Let $s_{\textup{same}}>GST$ be defined as in \textup{Lem.~\ref{lem: same scc}}. Then after finishing any slot $s\ge s_{\textup{same}}$, all correct nodes finalize the digest for slot $s-2$.
\end{theorem}
\begin{proof}
By Lem.~\ref{lem: same scc}, all correct nodes have adopted the same backbone chain before starting slot $s_{\text{same}}$. Then they produce the same digest for slot $s_{\text{same}}-1$ by Lem.~\ref{lem: same sc} and finalize the digest for slot $s_{\text{same}}-2$ by Lem.~\ref{lem: tangles}. The same arguments can be applied for any other slot $s> s_{\text{same}}$.
\end{proof}

\subsection{Proof of Theorem \ref{th: confirmation}}
\subsubsection{Liveness}
First, we show the liveness property of fast-path confirmation. By Lem.~\ref{lem: same scc}, at some moment $s_{\text{same}}$ after $GST$ all correct nodes end up on the same backbone chain. Then any cautious honest transaction $tx$, included in a block $B$ after slot $s\ge s_{\text{same}}$, gets confirmed by $\ConfirmTransactions()$ at the state update phase of round $\BS{B}{time}+3$. Indeed all correct nodes stay on the same backbone chain by Lem.~\ref{lem: same sc} and add blocks created by correct nodes at one round at the state update phase of the next round. Since the transaction is honest and cautious, the blocks of correct nodes at round $\BS{B}{time}+2$ serve as a quorum of transaction certificates for $tx$ in $B$, resulting in confirmation of $tx$ in procedure $\ConfirmTransactions()$.

Second, we prove the liveness property of the consensus-path confirmation. Any transaction included in a block $B$ will be processed in the procedure \\ $\FinalizeTransactions()$ (see Alg.~\ref{alg: ledger}), when the finality time of slot $\BS{B}{slot}$ is determined. Due to Lem.~\ref{lem: same scc}, the finality time progresses after $GST$, specifically, after slot $s_{\text{same}}$. If a cautious honest transaction has a transaction certificate in the DAG (restricted to the respected digest), then it will be added to the final ledger in the first for-loop of $\FinalizeTransactions()$; if not, it will happen in the second for-loop. For a cautious transaction, its inputs are already in $\LOG$; for an honest transaction, no conflicting transactions could appear in $\LOG$ before processing this transaction.

\subsubsection{Consistency}
Let $tx$ be a transaction in a block $B$ such that $tx$ in $B$ gets confirmed through the fast path by a correct node. Hence, there exists a quorum of transaction certificates for $tx$ in $B$, where each transaction certificate is created at slot $\BS{B}{slot}$ or $\BS{B}{slot}+1$. 

Next, we show that every correct node adds $tx$ to $\LOG$ through the consensus path. Let the finality time for $\BS{B}{slot}$ be $\tau$, i.e., $\FinalTime(\BS{B}{slot}) = \tau$, see Def.~\ref{def: final time}. After the moment when a correct node updates the finality time with $\tau$, this node will call $\FinalizeTransactions()$ at line~\ref{proc: finalizeTransactions} in Alg.~\ref{alg: procedures CreateBlock UpdateChain FinalizeSlot}. When processing transaction $tx$ in $B$ in the first for-loop (see line~\ref{alg: processed with transaction certificate}) in $\FinalizeTransactions()$, there will be at least one transaction certificate TC in $\D(\sigma_{\tau})$. Indeed, one correct node must contribute to both a quorum of transaction certificates for $tx$ in $B$ and a quorum of digest certificates when finalizing digest $\sigma_{\tau-2}$ at slot $\tau$. By definition of approvals (Def.~\ref{def: transaction approval}), the transactions creating inputs of $tx$ must be already in $\LOG$ and, thus, the if-condition at line~\ref{alg: inputs there} will be passed.
It remains to check that $tx$ is not conflicting with $\LOG$. 

Toward a contradiction, assume that a conflicting transaction $tx'$ is already in $\LOG$. First, recall that both $tx$ and $tx'$ can not have transaction certificates due to quorum intersection by at least one correct node. Thus $tx'$ had a chance to get confirmed only when it was processed with block $B'$  in the second for-loop (see line~\ref{alg: processed with total order}) in $\FinalizeTransactions()$ at some point before, i.e., when the partition of the DAG corresponding to one of the previous finality times $\tau'$ with $\tau'< \tau$ was processed. There should be at least one correct node that contributed to both the finalization of digest $\sigma_{\tau'-2}$ (which commits block $B'$ with $tx'$) at slot $\tau'$ and approval of $tx$ in block $B$. However, this could potentially happen only if $\BS{B}{slot}\le \tau'-1$ because of the definition of transaction approval (see Def.~\ref{def: transaction approval}). This means both slots $\BS{B}{slot}$ and $\BS{B}{slot} +1$ are at most $\tau'$, and there should be at least one correct node with a transaction certificate for $tx$ in $B$ and contributing to the finality of $\sigma_{\tau' -2}$ at slot $\tau'$. However, this means that $tx$ would get final when processing the first loop (see line~\ref{alg: processed with transaction certificate} in $\FinalizeTransactions()$) before $tx'$ in $B'$, i.e., $tx'$ can not be added to $\LOG$.

\subsubsection{Safety}
Due to quorum intersection by at least one correct node, two conflicting transactions can not have transaction certificates and can not get both to $\LOG$ through the fast path. When adding a transaction to $\LOG$ in procedure $\FinalizeTransactions()$,  we check a potential double-spend for the transaction in the current state of the ledger. 

It remains to check the impossibility of the remaining case when a transaction $tx$ is added to $\LOG$ through the fast path and at the time of updating the ledger, $tx$ is conflicting with some transaction $tx'\in\LOG$. Such a check is absent in procedure $\ConfirmTransactions()$. However, the same case was considered in the proof of consistency above, and its impossibility was already shown.

\section{Conclusion and discussion}
This paper has introduced Slipstream, a DAG-based consensus protocol with an integrated payment system. Slipstream provides two block orderings, which achieve ebb-and-flow properties~\cite{neu2021ebb}: the optimistic (or available) ordering is designed for the slot-sleepy model, where communication is synchronous, and each node is either awake or asleep during each slot; the final ordering is intended for the eventually lock-step synchronous model, where communication is initially asynchronous before GST and becomes lock-step synchronous after GST. Notably, the protocol achieves deterministic safety and liveness in a sleepy model, which was not previously demonstrated. However, this is achieved through stricter requirements in our sleepy model, including its lock-step nature and stable participation of awake nodes during $f+2$ rounds. The payment system allows for efficient UTXO transaction confirmation: under synchrony, transactions can be confirmed after 3 rounds and unconfirmed double-spends are resolved in a novel way that utilizes the final block ordering in the DAG structure. We have demonstrated how this technique can be applied to other payment systems that use a DAG for disseminating transactions.

\textbf{Latency:} The latency to commit a block to the optimistic and final orderings in the respective models is $\Theta(f)$ rounds. The high latency for the eventually lock-step synchronous model results from the high latency in the slot-sleepy model. It would be interesting to explore how to improve our proposed DAG-based protocol to allow committing blocks in a sleepy model after $O(1)$ rounds (on average). For instance, one could attempt to commit leader blocks for optimistic block ordering, where leaders are chosen based on VRF evaluations, as done in~\cite{momose2022constant}.

\textbf{Randomization:} Slipstream uses randomization in the ELSS model to enable all correct nodes to safely switch their backbone chains to the one proposed by the leader. We have shown that $O(n)$ slots are sufficient for this synchronization. It would be interesting to investigate how to integrate a leaderless chain-switching rule into Slipstream that does not use any randomization and improves the average number of attempts required for successful synchronization.

\textbf{Communication complexity: }The amortized communication complexity in Slipstream is shown to be $O(n^2)$. While it is unclear how to achieve $O(n)$ complexity in sleepy models, one could attempt to optimize this for the eventually lock-step synchronous model. A natural solution for this is employing error-correcting codes, as is done for Byzantine reliable broadcast primitives~\cite{cachin2005asynchronous,das2021asynchronous}. Transaction dissemination in blocks could be optimized by encoding them using Reed-Solomon codes, such that any $f+1$ chunks from the $n$ chunks of encoded transaction data are sufficient to reconstruct the original transaction data. Instead of sending all transaction data from past blocks, each node could send only one chunk corresponding to its index. In this way, it may be possible to achieve linear amortized communication complexity when the transaction data is large enough.

\section{Acknowledgement}
Mayank Raikwar has been supported by IOTA Ecosystem Development grant.

\newpage

\bibliographystyle{plainurl}
\bibliography{ref}
\appendix


    

    
\newpage
\section{Valid DAG}\label{sec: valid DAG}
The concept of a valid DAG is needed to capture the essential properties of a DAG that is maintained by a correct node during the execution of Slipstream in one of the models considered in the paper. The following three sections delve into different definitions of valid DAGs, whereas the last one shows why a correct node always maintains a valid DAG.
\subsection{Graph-valid DAG}
A DAG $\D=(V,E)$ is considered \textit{graph-valid} if it satisfies the following properties:
\begin{enumerate}
    \item the vertex set $V$ contains a unique vertex, denoted as $Genesis$, with no outgoing edges, i.e., which does not contain any hash references;
    \item for any vertex $B\in V$, the set $V$ contains all blocks that are referenced by $B$ in $\BS{B}{refs}$;
    \item for any two vertices $B$ and $C$ with $C\in \BS{B}{refs}$, the edge set $E$ contains the directed edge $(B,C)$.
\end{enumerate}
\subsection{Time-valid DAG}
 A DAG $\D$ is called \textit{time-valid} if it satisfies the following properties:
\begin{enumerate}
     \item for any two vertices $B,C\in \D$ with $C\in \BS{B}{refs}$, it holds $\BS{C}{time}<\BS{B}{time}$;
     \item there is only one vertex $Genesis\in \D$, that is created at slot $0$; specifically, $Genesis.time = \langle 0,f+2\rangle$.
\end{enumerate}
\subsection{Digest-valid DAG}\label{sec: digest correct}
A DAG $\D$ is called \textit{digest-valid} if it satisfies the following digest validity (DV) properties:
\begin{enumerate}
    \item the digest of genesis $\BS{Genesis}{digest} = 0$; \label{item: CV1}
    \item every block $B\in \D$  with $\BS{B}{round}< f+2$ contains a digest to slot $\BS{B}{slot}-2$, i.e., $\BS{B}{digest}\mathrm{.slot} = \BS{B}{slot}-2$. If $\BS{B}{round} = f+2$, $B$ contains digest to slot $\BS{B}{slot} -1$; \label{item: CV2}
    \item for any block $B \in \D$, $B\neq Genesis$ it holds
    \begin{enumerate}
        \item if $\BS{B}{round}\not\in\{ 1, f+2\}$, all referenced blocks $C\in \BS{B}{refs}$ must satisfy $\BS{C}{digest}=\BS{B}{digest}$;
        \item if $\BS{B}{round} = 1$,  all blocks in $\BS{B}{refs}$ can be divided into two groups $\mathrm{Ref}_1$, $|\mathrm{Ref}_1|\ge 1$, and $\mathrm{Ref}_2$, $|\mathrm{Ref}_2|\ge 0$,  such that for any $C\in \mathrm{Ref}_1$, $\BS{C}{digest}=\BS{B}{digest}$ and for every $F\in \mathrm{Ref}_2$ the digest field is the same (but different from $\BS{B}{digest}$);
        \item if $\BS{B}{round} = f+2$, all blocks $C\in \BS{B}{refs}$ referenced by $B$ must satisfy $\BS{C}{digest}=\BeforeCommit(\BS{B}{digest})$; see Def.~\ref{def: commitment};
    \end{enumerate} \label{item: CV3}
    \item a digest of every block $B\in\D$ with $\BS{B}{round}=f+2$ is computed correctly based on $\cone(B)$ according to Def.~\ref{def: commitment}. \label{item: CV4}
\end{enumerate}
\subsection{Valid DAG of correct nodes}
\begin{definition}[Valid DAG]\label{def: validDAG}
A DAG is called \textit{valid} if it is graph, time, and digest-valid. 
\end{definition}

\begin{lemma}\label{lem:DAGisValid}
For any given correct node, it holds that its local DAG $\D$ at the end of the state update phases is valid. 
\end{lemma}
\begin{proof}
Recall that the local DAG $\D$ at the end of the state update coincides with the past cone $\cone(B_{\text{own}})$ of the latest created block. This is straightforward to check for graph- and time-validity. In the remainder, we focus on digest-validity.

There are four different cases in the Alg.~\ref{alg: consensus} to check. The first case is when round $i \notin \{1,f+2\}$ and $\UpdateDAG()$ is used. Note that before adding $\cone(B)$ to the local DAG in line~\ref{alg: add cone}, we check the validity of $\cone(B)$, see Def.~\ref{def: validDAG}. It remains to verify that the union of two valid DAGs is again a valid DAG. Conditions DV.\ref{item: CV1}, DV.\ref{item: CV2}, and DV.\ref{item: CV4} follow also directly from the fact that $\cone(B)$ and $
\D$ do satisfy them. Condition DV.\ref{item: CV3} is satisfied since we require in line \ref{alg: requireSameCommit} that $B$ has the same digest as the local DAG. As the node is supposed to be honest and to follow the protocol, the new block $B_{\text{own}}$, created in line \ref{alg: bown created} in Alg.~\ref{alg: consensus}, only references blocks in $\D \cup \cone(B)$ which have the same digest and hence $\cone(B_{\text{own}})$ is valid.

The second case is where $i=1$ and the node wakes up. 
Similar to the arguments in the $\UpdateDAG()$, the validity properties are conserved by merging two valid DAGs. Note hereby that property DV.\ref{item: CV3} is true since we only accept blocks $B$ having the same digest. 

The third case, where $i=1$ and the node switches the chain, is treated analogously. Let us only note here that, for property DV.\ref{item: CV3}, we have to check that the new block $B_{\text{own}}$, satisfies  DV.\ref{item: CV3}.b. However, this condition is satisfied, as at most one block with a different digest is added. 

The fourth case is where $i=f+2$ and the backbone chain gets updated. This has no influence on graph and time validity; however, the new block $B_{\text{own}}$ will contain a different digest. For an honest node following procedure \UpdateChain, see line \ref{proc:UpdateChain} in Alg.~\ref{alg: procedures CreateBlock UpdateChain FinalizeSlot}, properties DV.\ref{item: CV2},  DV.\ref{item: CV3}, and DV.\ref{item: CV4} hold true.
\end{proof}


\newpage

\section{Algorithms}\label{sec: algorithms}

\begin{algorithm}[!bht]
\scriptsize
\caption{Procedure to update the DAG.}
\label{alg: updateDAG}
\BlankLine
\Loc{}{
$\D$ \tcp*{DAG maintained by the node}
$\buffer$ \tcp*{Buffer maintained by the node}
$\EqSet$ \tcp*{Set of equivocators known to the node}
$B_{\text{own}}$ \tcp*{Last block created by the node}
}
\Proc{$\UpdateDAG()$\label{proc: updatedDAG}}{
$\langle s, i\rangle \gets \Now()$ \;
$\langle s', i'\rangle \gets \BeforeTime(\langle s, i\rangle)$ \;
\If{$B_{\textup{own}}.\mathrm{time} \neq \langle s', i'\rangle$}{\Return}
$\sigma\gets \CurrentCommit()$ \tcp*{see Line~\ref{line:currencommit} of Alg.~\ref{alg: procedures CreateBlock UpdateChain FinalizeSlot}} 
\For{$ B\in \buffer$ \textup{s.t.} $\cone(B)\subset \buffer$ $\land$ $\BS{B}{time}=\langle s',i'\rangle$ $\land$ $\BS{B}{digest}=\sigma$ $\land$ $\BS{B}{node}\not\in \EqSet$}  { \label{alg: requireSameCommit}
$check\gets 1$  \;
\For{$ C\in \cone(B)$ \textup{s.t.} $\BS{C}{slot}=s$}{
\If{$\BS{C}{slot}\in \EqSet(\sigma)$\label{alg: check equivocators}}{
$check \gets 0$ \;
\textbf{break}
}
}
\For{$C\in \cone(B)\setminus \D$ \textup{s.t.} $\BS{C}{slot} \le s-1$ $\land$ $\sigma$ \textup{does not commit} $C$}{
\If{$\ReachNumber(C, B)< i-1$\label{alg: reachnumber}}{
$check \gets 0$ \;
\textbf{break}
}
}
\If{$check$ $\land$ $\IsValid(\cone(B))$}{
$\D \gets \D \cup \cone(B)$ \label{alg: add cone} \;
}
}
}
\tcp{The function outputs number of nodes that created a block in $\cone(B)$ at slot $\BS{B}{slot}$ from which one can reach $C$}
\Func{$\ReachNumber(C,B)$}{
$S\gets \{\}$ \tcp*{Set of nodes}
\For{$E\in \cone(B)$ \textup{s.t.} $\BS{E}{slot}=\BS{B}{slot}$}{
\If{$C\in \cone(E)$}{
$S\gets S \cup \{\BS{C}{node}\}$ \;
}
}
\Return $|S|$
}
\tcp{The function outputs the previous timestamp}
\Func{$\BeforeTime(\langle s,i\rangle)$}{
\eIf{$i>1$}{\Return $\langle s,i-1\rangle$}{\Return $\langle s-1,f+2\rangle$}
}

\tcp{The function verifies if a DAG is graph-, time-, and digest valid}
 \Func{\IsValid{$\mathcal{C}$}} { \label{line: IsValidDag}
\Return $\mathcal{C} \textup{ is valid}$ \tcp*{see Def.~\ref{def: validDAG}}

}
\end{algorithm}

\begin{algorithm}[!htbp]
\scriptsize
\caption{Procedures to create a new block, update a backbone chain, and finalize a digest.}
\label{alg: procedures CreateBlock UpdateChain FinalizeSlot}
\BlankLine

\Loc{}{
$\D\gets \{Genesis\}$ \tcp*{DAG maintained by the node}
$\chain_{\text{own}}\gets (0)$ \tcp*{backbone chain adapted by the node}
$\text{pk}, \text{sk}$ \tcp*{Public key and Private key of the node}
$s_{\text{final}}\gets 0$\tcp*{Slot index of the last final slot digest}
$s_{\text{pre}}\gets 0$ \tcp*{Slot index of the last final slot digest in $\D(s_{\text{final}})$, see Def.~\ref{def: final time}}
$B_{\text{own}}$ \tcp*{Last block created by the node}
$\Order_{\text{own}}, \Order_{\text{final}}$\tcp*{Optimistic and final orders of blocks}
}

\Func{\CreateBlock{} \label{fn: create block}}{
$B\gets \{\}$\;
$B.\mathrm{refs}\gets \textsc{Hash}(\textsc{Tips}())$ \tcp*{Hash references to all unreferenced blocks in $\D$}
$\BS{B}{txs}\gets \textsc{Payload}()$ \tcp*{Subset of transactions}
$\BS{B}{digest} \gets \CurrentCommit()$\;
$\BS{B}{time}\gets \textsc{Now}()$ \tcp*{Current slot and round indices}
$\BS{B}{node} \gets \text{pk}$ \;
$\BS{B}{sign} \gets \textsc{Sign}_{\text{sk}}(B)$ \tcp*{Sign the content of block}
$\D\gets \D\cup \{B\}$ \tcp*{Update the maintained DAG}
\Return $B$
}

\Func{\Tips{}}{
\Return $\{ B\in \D:\; \nexists C\in \D: \Hash(B)\in \BS{C}{refs} \}$
}
\Func{\CurrentCommit{}}{ \label{line:currencommit}
\Return $\chain_{\text{own}}.last$ \tcp*{Last element of backbone chain}
}

 \Func{\IsQuorum{$Set$}} { \label{line:quorum}
\Return $|\BS{B}{node}: \ B\in Set| \ge 2f+1$
}
\Proc{\UpdateChain{}}{ \label{proc:UpdateChain}
$\langle s+1,f+2 \rangle \gets \Now()$ \;
$\sigma_{s-1} \gets \CurrentCommit()$ \;
$r\gets \{\}$\;
$Blocks_{\le s} \gets \{B\in \D\setminus \D(\sigma_{s-1}) : \ \BS{B}{slot} \le s\}$\;
\For{$ B\in Blocks_{\le s}$}{
    $r\gets r || \Hash(B)$\;
}
$\sigma_{s} = \Hash(\sigma_{s-1}, r)$ \;
$\chain_{\text{own}} = \chain_{\text{own}} || \sigma_{s} $ \tcp*{Append the computed digest to the end}
$\Order_{\text{own}}\gets \Order(\sigma_s)$ \;
} 
\Proc{\FinalizeSlots{}}{ \label{proc:FinalizeSlots}
$\langle s,i \rangle \gets \Now()$ \; 
\For{$ t \in [s_{\textup{final}}+3, s]$}
{$\sigma_{t}\gets\chain_{\text{own}}[t]$\;
$Blocks_{=t}\gets \{B\in \D(\sigma_{t}) : \ \BS{B}{slot} = t\}$\;
DC$_{=t}\gets \{B\in Blocks_{=t}:\ B$  is a DC for $\sigma_{t-2} \}$\tcp*{ DCs from slot $t$ (see Def.~\ref{def: quorum certificate})}
\If{$\IsQuorum(\textup{DC}_{=t})$}{
    $s_{\textup{final}}\gets t-2$\;
    $\Order_{\text{final}} \gets \Order(\sigma_{t-2})$ \;
     $\tau\gets \FinalTime(s_{\text{pre}}+1)$ \tcp*{see Def.~\ref{def: final time}}
     \While{$\tau\neq \bot$}{
     $\sigma\gets \chain_{\text{own}}[\tau]$\;
    $s_{\text{pre}}\gets \LastFinal(\D(\sigma))$ \tcp*{Slot index of last final digest in $\D(\sigma)$}
    $\FinalizeTransactions()$ \label{proc: finalizeTransactions} \tcp*{see Alg.~\ref{alg: ledger}}
    $\tau\gets \FinalTime(s_{\text{pre}}+1)$\;
    }
}
}
}
\end{algorithm}

\begin{algorithm}[!htbp]
\scriptsize
\caption{Procedure to switch a chain.}
\label{alg: chain}
\BlankLine

\Loc{}{
$\D$ \tcp*{DAG maintained by the node}
$\chain_{\text{own}}$ \tcp*{backbone chain adapted by the node}
$B_{\text{own}}$ \tcp*{Last block created by the node}
$I_{\text{ELSS}}\gets 0$ \tcp*{Indicator that detects an ELSS model}
$s_{\text{final}}$ \tcp*{Slot index of the last final slot digest}
$\EqSet$ \tcp*{Set of equivocators known to the node}
}
\tcp{Procedure to adopt a chain when the node was slot-$s$ awake}
\Proc{\SwitchChain{} \label{proc:SwitchChain}}{
    $\langle s+1, 1 \rangle \gets \Now()$\;
    $\sigma_{s-1}\gets \CurrentCommit()$\;
    $N_{\text{same}} \gets 0, N_{\text{total}} \gets 0$ \;
    \For{$ node \in \{1,2,\ldots, n\}\setminus \EqSet$} {
     $B\gets \LastBlockFrom(node)$ \;
     \If{$\BS{B}{time} = \langle s, f+2\rangle$}{
     $N_{\text{total}}\gets N_{\text{total}}+1$\;
      \If{$\BS{B}{digest} = \sigma_{s-1}$}{
        $N_{\text{same}}\gets N_{\text{same}}+1$\;
        }
     }
    }   
    $\CheckELSS()$ \label{alg: checkElss} \tcp*{see Alg.~\ref{alg: identifyELSS}}
    $\text{leader} \gets \Leader()$ \;
    \If{$s_{\textup{final}}\neq s-2$ $\land$ $\textup{leader} \neq \bot$ \label{alg: finalized before} \label{alg: has not finalized and leader}}{ 
    $B_{\text{leader}}\gets \LastBlockFrom(\text{leader})$ \;
    \If{$\neg$\IsValid{$\cone(B_{\textup{leader}})$}}{\textbf{break}}
    DC$_{\text{leader}}\gets \LastCommitCertificate(B_{\text{leader}})$ \tcp*{Last DC created by the block creator in block's causal history}
    DC$_{\text{own}}\gets \LastCommitCertificate(B_{\text{own}})$ \;
    \If{$\IsConflict(\Commit(\textup{DC}_{\textup{leader}}), B_{\textup{own}}.\mathrm{digest})$}{
    $I_{\text{ELSS}}\gets 1$ \label{alg: change ELSS 1}\;
    }
    \eIf{$2N_{\textup{same}}\le N_{\textup{total}}$ }{
        \If{$\neg \IsConflict( \BS{B_{\textup{leader}}}{digest},\Commit(\textup{DC}_{\textup{own}}))$ $\lor$  \textup{DC}$_{\textup{leader}}.\mathrm{slot} \ge \textup{DC}_{\textup{own}}.\mathrm{slot}$ \label{alg: sneak to non-conflic}}{
    $\D \gets \D \cup \cone(B_{\text{leader}})$\;
        $\chain_{\text{own}} \gets \chain( \BS{B_{\text{leader}}}{digest})$ \;
        $\Order_{\text{own}} \gets \Order(\BS{B_{\text{leader}}}{digest})$ \;
    }
    }{ 
    \If{$I_{\textup{ELSS}}=1$ $\land$ $\textup{DC}_{\textup{leader}}.\mathrm{slot} \ge \textup{DC}_{\textup{own}}.\mathrm{slot}$ \label{alg: adopt CC: SS safety}}{
        $\D \gets \D \cup \cone(B_{\text{leader}})$\;
        $\chain_{\text{own}} \gets \chain(\BS{B_{\text{leader}}}{digest})$ \;
        $\Order_{\text{own}} \gets \Order( \BS{B_{\text{leader}}}{digest})$ \;
        } 
    }
}

}
\end{algorithm}

\begin{algorithm}[!htbp]
\scriptsize
\caption{Procedures to ``wake up'' and indicate the communication model.}
\label{alg: identifyELSS}
\BlankLine
\Loc{}{
$\D$ \tcp*{DAG maintained by the node}
$\chain_{\text{own}}$ \tcp*{backbone chain adapted by the node}
$I_{\text{ELSS}}\gets 0$ \tcp*{Indicator that detects an ELSS model}
$\EqSet$ \tcp*{Set of equivocators known to the node}
}
\tcp{Procedure to adopt a chain when the node was slot-$s$ asleep} \label{proc:WakeUpChain}
\Proc{\WakeUpChain{}}{
    $\langle s+1, 1 \rangle \gets \Now()$\;
    $M \gets \{\}$ \tcp*{Multiset of digests}
    \For{$ node \in \{1,2,\ldots, n\}\setminus \EqSet$} {
     $B\gets \LastBlockFrom(node)$ \;
     \If{$\BS{B}{time} = \langle s, f+2\rangle$}{    
      $M\gets M \cup \{\BS{B}{digest}\}$ \;
     }     
    } 
    $\sigma_{s-1}\gets \Mode(M)$ \tcp*{Most present element in the multiset with deterministic tiebreaker} 
    $\chain_{\text{own}} \gets \chain(\sigma_{s-1})$ \;
    $\Order_{\text{own}} \gets \Order(\sigma_{s-1})$ \;
    \For{$node \in \{1,2,\ldots, n\} \setminus \EqSet$} {
     $B\gets \LastBlockFrom(node)$ \;
     \If{\IsValid{$\cone(B)$} $\land$ $\BS{B}{digest} = \sigma_{s-1}$}{    
      $\D \gets \D \cup \cone(B)$ \;
     }    
     }
}

\tcp{Procedure to update the variable indicating the ELSS model} \label{proc:CheckELSS}
\Proc{\CheckELSS{}}{
$\langle s+1, 1 \rangle \gets \Now()$\;
    $M \gets \{\}$ \tcp*{Multiset of digests}
    \For{$ B \in \buffer$ \textup{s.t.} $\BS{B}{time} = \langle s-1,f+2\rangle $} {
     $M\gets M \cup \{\BS{B}{digest}\}$ \;         
    }
    \If{$\exists \sigma, \theta \in M$ \textup{s.t.} $\Number(\sigma, M) \ge f+1$ $\land$ $\Number(\theta, M) \ge f+1$}{
    $I_{\text{ELSS}} \gets 1$ \label{alg: change ELSS 2}\;
    }
    
    }
\tcp{The function outputs the number of repetitions of an element in a multiset}
\Func{$\Number(\sigma, M)$} {
$count\gets 0$ \;
\For{$ \beta\in M$}{
\If{$\beta = \sigma$}{
$count\gets count +1$\;
}}
\Return $count$}
    
\end{algorithm}

\begin{algorithm}[!htbp]
\caption{Procedures to update the ledger.}
\label{alg: ledger}
\scriptsize
\BlankLine
\Loc{}{
$\D$ \tcp*{DAG maintained by the node}
$\LOG$ \tcp*{Confirmed ledger maintained by the node}
$s_{\text{final}}\gets 0$\tcp*{Slot index of the last final slot digest}
$s_{\text{pre}}\gets 0$ \tcp*{Slot index of the last final slot digest in $\D(s_{\text{final}})$, see Def.~\ref{def: final time}}
$\ProcessedWithTxCert \gets \{\}$ \tcp*{Processed blocks with transactions for which transactions certificates were checked}
$\ProcessedTotalOrder \gets \{\}$ \tcp*{Processed blocks with transactions which were resolved using the total order}
}
\tcp{Procedure to confirm transaction through the consensus path}
\Proc{$\FinalizeTransactions()$}{
$\tau\gets \FinalTime(s_{\text{pre}})$ \tcp*{see Def.~\ref{def: final time}}
\tcp{First add transactions with transaction certificates (see Def.~\ref{def: transaction certificate}). Use $\Order_{\text{final}}$ for the for-loop.}
\For{$ B\in \D(\sigma_{\tau})\setminus \ProcessedWithTxCert$ \textup{s.t.} $\BS{B}{slot}\le \tau-2$ \label{alg: processed with transaction certificate}}{
$\ProcessedWithTxCert\gets \ProcessedWithTxCert \cup \{B\}$ \;
 \For{$ tx \in \BS{B}{txs}$} {
    \If{$\exists C\in \D(\sigma_{\tau})$ \textup{s.t.} $C$ \textup{is a} \textup{TC for} $tx$ \textup{in} $B$} {
     \If{$tx$ \textup{has inputs in} $\LOG$} {
     
     \If{$tx$ \textup{is not conflicting with} $\LOG$}{
     $\LOG \gets \LOG \cup \{tx\}$\;  
     }
    }
    
    }
    }    
}
\tcp{Use total order to resolve remaining conflicts}
\For{$ B\in \D(\sigma_{\tau-2})\setminus \ProcessedTotalOrder$ \label{alg: processed with total order}}{
$\ProcessedTotalOrder\gets \ProcessedTotalOrder \cup \{B\}$ \;
 \For{$ tx \in \BS{B}{txs}$} {
    \If{$tx$ \textup{has inputs in} $\LOG$\label{alg: inputs there}} {
     
     \If{$tx$ \textup{is not conflicting with} $\LOG$}{
     $\LOG \gets \LOG \cup \{tx\}$\;  
     }
    }
    
    }    
}
}

\tcp{Procedure to confirm transaction through the fast-path}
\Proc{$\ConfirmTransactions()$ \label{proc: confirmTransactions}}{
\For{$ B \in \D$}{
\For{$ tx\in \BS{B}{txs}$}{
TC$_{\textup{all}}\gets \AllTxCertificates(tx, B)$ \; 
\If{$\IsQuorum(\textup{TC}_{\textup{all}})$}{$\LOG\gets \LOG \cup \{tx\}$\;
}
}
}
}
\Func{$\AllTxCertificates(tx,B)$}{
$M\gets\{\}$ \tcp*{Multiset of transaction certificates}
\For{$ C \in \D$ s.t. $0\le \BS{C}{slot}-\BS{B}{slot}\le 1$}{
\If{$C$ \textup{is a TC for} $tx$ \textup{in} $B$}{
$M\gets M\cup \{C\}$\;
} 
}
\Return $M$
}
\end{algorithm}
\section{Comparison with DAG-based BFT Protocols}\label{sec: comparison with dag-based bft protocls}
\textbf{DAG-based BFT protocols:}
There is a class of DAG-based consensus protocols with threshold clocks designed for partially synchronous and asynchronous networks.  In a round-based protocol, a node can increase the round number only when the DAG contains a quorum of blocks with the current round number.  These protocols optimize their performance by assigning leaders for certain rounds and using special commit rules for leader blocks. Committed leader blocks form a backbone sequence that allows for the partitioning of the DAG into slices and the deterministic sequencing of blocks in the slices. 

Many protocols in this class such as Aleph~\cite{gkagol2018aleph}, DAG-Rider~\cite{keidar2021all},  Tusk~\cite{danezis2022narwhal}, and Bullshark~\cite{spiegelman2022bullshark} use special broadcast primitives such as Byzantine Reliable Broadcast (BRB) and Byzantine Consistent Broadcast (CRB) to disseminate \textit{all} blocks. Therefore, during the execution of any of these protocols, nodes construct a \textit{certified DAG}, where each block comes with its certificate, a quorum of signatures, which does not allow any equivocating blocks to appear in the DAG. The latency for non-leader blocks in the above protocols suffers from introducing \textit{waves}, a number of consecutive rounds with one designated leader block. Shoal~\cite{spiegelman2023shoal} reduces the latency of non-leader blocks by interleaving two instances of Bullshark.  


All the above protocols use broadcast primitives for every block, leading to an increased latency compared to the state-of-the-art chain-based consensus protocols like HotStuff~\cite{malkhi2023hotstuff,yin2019hotstuff}.  To reduce this latency and simplify the consensus logic, BBCA-Chain~\cite{malkhi2023bbca} suggested using a new broadcast primitive, called Byzantine Broadcast with Complete-Adopt (BBCA), only for leader blocks, and Best-Effort Broadcast (BEB) for all other blocks. In addition, BBCA-Chain makes all rounds symmetric by assigning leaders every round. Note that there has been a similar effort; specifically, Sailfish~\cite{shrestha2024sailfish} and Shoal++~\cite{arun2024shoal++} assign a leader node for every round and allow committing even before BRB instances deliver voting blocks. This allows improving the latency in the \textit{happy-case} when all nodes are correct or in case of crashed nodes.

 Hashgraph~\cite{baird2016swirlds} is the first DAG-based leaderless BFT protocol in which nodes use BEB for disseminating all their blocks. Nodes construct an \textit{uncertified} or \textit{optimistic} DAG, in which equivocating blocks could appear. By logical interpretation of the DAG, nodes exclude equivocations and run an inefficient binary agreement protocol that orders blocks by the received median timestamps. This leads to very high latency in the worst case. Cordial Miners~\cite{keidar2022cordial} use a similar mechanism to exclude equivocations. However, the latency is significantly improved by assigning leaders in each 3-round wave and introducing a special commit rule for leader blocks on such an optimistic DAG. Mysticeti~\cite{babel2023mysticeti} is the most recent improvement of Cordial Miners that introduces pipelined leader blocks. Mysticeti allows every node at every round to be a leader, which significantly reduces the latency, especially in case of crash nodes.

 Our protocol, Slipstream, uses BEB to disseminate all blocks in the DAG and uses time slots to partition the DAG into slices. On the one hand,  the happy-case latency of committing blocks in slices is linear with the number of Byzantine nodes, which is much higher than the constant average latency in many other DAG-based solutions. On the other hand, all existing DAG-based BFT protocols halt in case of network partitions (e.g., when $1/3$ of nodes are offline). In contrast, our protocol allows for optimistic committing the DAG in such cases and tolerating a higher fraction of Byzantine nodes under a synchronous slot-sleepy network model.

\section{Adapting Slipstream's Consensus Path for UTXO transactions to Mysticeti-FPC}\label{sec: adapting Slipstream to Mysticeti-FPC}
\subsection{Basics on Mysticeti-C and Mysticeti-FPC}
First, we recall how the Mysticeti-C protocol works and refer to~\cite{babel2023mysticeti} for more details. The causal history of block $L$ is denoted as $\cone(L)$. Each block from round $r$ has to reference at least $2f+1$ blocks from round $r-1$. A block $X$ votes for a block $Y\in \cone(X)$ if  $\BS{Y}{round} = \BS{X}{round}-1=r$ and $Y$ is the first block from that block creator $\BS{Y}{node}$ in that round $r$ in the depth-first traversal of $X$. A block $X$ is called a block certificate for $Y$ if $\BS{Y}{round}=\BS{X}{round}-2$ and $\cone(X)$ contains a quorum of blocks voting for $Y$. 
All nodes agree on the scheduler of \textit{leader slots}, where a leader slot is a pair of a designed node and a round number. A block $L$ from a leader slot is  a leader block. A leader block is called committed in two cases:
\begin{enumerate}
    \item \textbf{Directly committed:} a local DAG contains a quorum
of block certificates for $L$;
    \item \textbf{Indirectly committed:} 2.1) a local DAG contains a directly committed leader block $B$ from round at least $\BS{L}{round}+3$ that has a path (on the local DAG) to a block certificate for $L$, and 2.2) all earlier (than $B$) committed leader blocks have no path (on the local DAG) to a block certificate for $L$.
\end{enumerate}

Next, we recall the transaction confirmation rule in Mysticeti-FPC. For the sake of simplicity, we consider only owned-object transactions. Nodes include blocks with transactions. In addition, each block contains explicit votes for causally past transactions. Each correct node is assumed to vote for each transaction once. A block is called a transaction certificate for a given transaction if it can reach blocks from $2f + 1$ nodes, including a vote for the given transaction. A given transaction is called confirmed if one of the two conditions holds:
\begin{enumerate}
    \item \textbf{Fast path:} the local DAG contains $2f+1$ transaction certificates for the given transaction;
    \item \textbf{Consensus path:} the causal history of a committed leader block contains (at least) one transaction certificate for the given transaction; 
\end{enumerate}

It is shown in~\cite{babel2023mysticeti} that every transaction confirmed through the fast path is also confirmed through the consensus path.

\subsection{Adapting the consensus path in Mysticeti-FPC}
We will describe one possible variation for the consensus path that enables unconfirmed double spends to be resolved. We propose to make votes in blocks only for transactions included in blocks from recent rounds and add an additional option for confirmation in the consensus path. Specifically, after confirming transactions with a transaction certificate, one can take an \textit{old enough} part of the DAG and confirm unconfirmed transactions using the total order, even if in the absence of transaction certificates. 
However, we must ensure that any transaction confirmed through this method remains consistent with the fast-path confirmation, such that no conflicts arise between transactions confirmed via the consensus path and those confirmed via the fast path. 

We redefine the concept of a transaction certificate to align with the time constraints used for block certificates, as follows (changes are highlighted in \textcolor{blue}{blue}).

    \begin{definition}
         A block $X$ is called a transaction certificate for a transaction $tx$ in block $Y$ if \textcolor{blue}{$\BS{Y}{round}=\BS{X}{round}-2$ and} $\cone(B)$ contains blocks from $2f + 1$ nodes, including a vote for the transaction $tx$ in $B$.
    \end{definition}
    \begin{remark}
    We note that a block certificate in Mysticeti-C has a similar condition on the round of voting blocks, whereas a transaction certificate in Mysticeti-FPC can appear anywhere in the causal future of a transaction.
    \end{remark}
Apart from assuming that each correct node votes for each transaction once, we also assume that once a correct node detects double spending transaction, it does not include a vote for any of them.

We will partition the DAG into slices in a deterministic way using directly committed leader blocks. 
In the protocol, nodes keep track of a backbone chain of committed leader blocks (both directly and indirectly), say $L_1, L_2, L_3, \ldots$. For each leader block $L_i$ in this backbone chain, define $C(L_i)$ as $L_j$ with minimum $j>i$ such that $L_i$ is committed in the DAG restricted to $\cone(L_j)$.  We denote by $D(L_i)$  the leader block $L_j$ with the minimum $j\ge i$ such that $L_j$ is directly committed in the DAG $\cone(C(L_i))$. 
Note that $\BS{C(L)}{round}\ge \BS{D(L)}{round}+3$. It will be important that $D(L)$ is directly committed, e.g., $D(L)$ is referenced by at least $2f+1$ blocks in the next round. A deterministic definition of $D(L)$ using $C(L)$ is needed to ensure that all correct nodes will make the same operations. For a sequence of committed leader blocks $L_1,L_2,\ldots$, one can define a sequence of distinct $C_1, C_2, \ldots $ and a sequence of distinct $D_1, D_2,\ldots $.

Now there are three ways for a transaction in a given block to become confirmed:
\begin{enumerate}
    \item \textbf{Fast path:} the local DAG contains $2f+1$ transaction certificates over the
block proposing the transaction;
    \item \textbf{Consensus path:}
    \begin{itemize}
        \item[2a] When a new leader block $C=C(L)$ is determined based on the committed backbone chain. We consider all blocks created at or before round $\BS{C}{round}-3$  and located in $\cone(C)$.  We check whether transactions in those blocks have at least one transaction certificate in $\cone(C)$. For every transaction with at least one transaction certificate, if possible (i.e., all input objects are confirmed and no double spend for this transaction was confirmed), we confirm the transaction; 
        \item[2b] When a new leader block $D=D(L)$ is determined based on the backbone chain of committed leader blocks. We traverse all transactions contained in blocks in $\cone(D)$. For every such transaction, if possible (all input objects are confirmed and no double spend for this transaction was yet confirmed), we confirm the transaction. At this step, we apply a total order over the transactions in the DAG $\cone(D)$. 
    \end{itemize}
\end{enumerate}
Consensus path 2b allows us to resolve unconfirmed double spends and confirm transactions that didn't get a certificate in time. Now we prove a statement showing the consistency between the fast path and the consensus path.
\begin{lemma}
A transaction confirmed through the fast path by a correct node will be confirmed through the consensus-path 2a by all correct nodes.
\end{lemma}
\begin{proof}
Assume $tx$ in block $B$ is fast-path confirmed. From quorum intersection, it follows that at least one transaction certificate will be in the causal history of a leader $C$ with $\BS{C}{round}>\BS{B}{round}+2$. It remains to show that no double spend for $tx$ could be confirmed before block $C$ is determined. Toward a contradiction, assume a double spend $tx'$ was confirmed before. Transaction $tx'$ can't be confirmed using options 1 and 2a due to the quorum intersection for transaction certificates. Still, we need to check that $tx'$ could not be confirmed when committing an earlier leader, say $D'=D(L')$, through consensus path 2b. Since $D'$ was directly committed and $tx$ was fast-path confirmed, there is at least one correct node contributing with its blocks for voting of $D'$ and $tx$. Approval of $tx$ could happen only strictly before the approval of $D'$ since otherwise, that correct node could not approve $tx$ while knowing about $tx'$ in $\cone(D')$. Thus, $\BS{B}{round} < \BS{D'}{round}$. However, $tx$ is fast-path confirmed with a quorum of transaction certificates being created at $\BS{B}{round}+2\le  \BS{D'}{round}+1$. Recall that $\BS{C(L')}{round} \ge \BS{D(L')}{round}+3$. This implies that at least one of the transaction certificates for $tx$ in $B$ was in $\cone(C(L'))$ and $tx$ should be consensus path confirmed through option 2a before transaction $tx'$ through 2b. This leads to a contradiction.
\end{proof}
Next, we informally analyze the latency for each of the paths for transaction confirmation. We assume a \textit{happy case} when, at one round, each correct node references blocks of all correct nodes from the previous round.  The fast-path confirmation of transactions optimistically could happen in three rounds (or two rounds after nodes receive a block with an owned-object transaction) similar to Mysticeti-FPC. The confirmation of transactions through the consensus path 2a happens once a transaction certificate appears on the DAG (3 rounds) and this certificate is included in the causal history of a committed leader block ($\ge 2$ rounds), i.e., it could happen in five rounds when this certificate is a leader block. The confirmation through the consensus path 2b for double spends that didn't get certificates happens when the block containing one of the double spends appears in $\cone(D(L))$, where $D(L)$ is a directly committed leader block derived using $C(L)$ and the sequence of all committed leader blocks; when such a transaction is in block $D$, then one needs at least four rounds to get in $\cone(C(L))$ and at least two more rounds to commit $C(L)$, i.e. it could happen in six rounds.

\end{document}